\documentclass[11pt]{article}
\usepackage[letterpaper,margin=1in]{geometry}
\usepackage[T1]{fontenc}
\usepackage{lmodern}
\usepackage{amsmath,amssymb,amsthm,mathtools,booktabs,array}
\usepackage{microtype,graphicx,placeins}
\usepackage[hidelinks]{hyperref}
\usepackage{cite}
\hypersetup{pdftitle={Quantum ordinal selection with two comparisons},pdfauthor={Hongyu Wu}}
\newcolumntype{L}[1]{>{\raggedright\arraybackslash}p{#1}}
\newtheorem{theorem}{Theorem}
\newtheorem{lemma}[theorem]{Lemma}

\newtheorem{corollary}[theorem]{Corollary}
\theoremstyle{definition}

\newcommand{\R}{\mathbb R}
\newcommand{\E}{\mathbb E}
\newcommand{\tr}{\operatorname{tr}}

\newcommand{\diag}{\operatorname{diag}}
\newcommand{\im}{\operatorname{im}}
\newcommand{\Cay}{\operatorname{cay}}
\newcommand{\mathscr}[1]{\mathcal{#1}}
\newcommand{\ind}{\mathbf 1}
\newcommand{\Q}{Q_2}
\newcommand{\PW}{P_W}
\title{Quantum ordinal selection with two comparisons}
\author{Hongyu Wu\\
\small Mike Wiegers Department of Electrical and Computer Engineering\\
\small Kansas State University, Manhattan, KS 66506, USA\\
\small \texttt{hongyuwu@ksu.edu}}
\date{}
\begin{document}
\maketitle

\begin{abstract}
We study the optimal probability of selecting an object within a prescribed rank threshold
using at most two noiseless coherent pairwise comparisons of an unknown strict order.
Polynomial degree bounds alone can overestimate the success attainable with two comparisons.
We develop a comparison-specific semidefinite model that completely characterizes attainable success.
Computer-assisted proofs establish the exact optimum for every rank threshold at seven objects and for minimum selection at eight.
For minimum selection, we construct an explicit algorithm family and determine its exact optimum.
Bounds for all admissible two-comparison algorithms show that the family's asymptotic success
is within 0.016\% of the unrestricted optimum.
We also prove that this relative gap has a strictly positive limit.
The optimal quantum-to-classical success ratio for minimum selection converges to approximately 2.48
under the same two-comparison budget.
\end{abstract}

\section{Introduction}\label{sec:intro}
With a fixed comparison budget, a selection procedure may have to return an object before its rank is known. Ordinal optimization provides a framework for selecting sufficiently good objects with a specified probability~\cite{HoSreenivasVakili1992,HoZhaoJia2007}. A blind pick returns a uniformly random label, so its chance of meeting the rank threshold is the accepted fraction of the population. Pairwise comparisons provide relative-order information that can improve this probability. Coherent comparisons also permit interference between queried pairs. In quantum ordinal optimization (QOO), the question is how quantum access changes selection success under a fixed information budget~\cite{wuqoo}.

The QOO formulation specifies the information model, budget and required output rank and gives an exact solution for one coherent comparison~\cite{wuqoo}. A second comparison lets an algorithm transform the state produced by the first query before querying the order again. The resulting question is how much this extra freedom can improve selection success under a two-comparison budget.

Prior quantum selection results establish query costs under several access models and success requirements. With value queries, a minimum can be found with constant success using a square-root number of queries~\cite{durrhoyer}. Minimum finding with bounded-error comparison answers is used in lexicographically minimal string rotation~\cite{wangying}. Approximate selection permits a specified rank tolerance, as in bounds for medians and other order statistics under value access and in the comparison-tree model~\cite{NayakWu1999}. Approximate value access also permits the construction of sets of smallest objects for every requested set size~\cite{gaojiwang}.

Ordered search instead assumes an already ordered list. Exact invariant algorithms exploit this promise for insertion~\cite{FarhiEtAl1999}. Bounded-error ordered search has a logarithmic query lower bound~\cite{Ambainis1999}, with further comparison-query bounds for searching and sorting~\cite{HoyerNeerbekShi2002}. Work on adversary weights~\cite{ChildsLee2008} and Bayesian updates~\cite{BenOrHassidim2008} also studies this ordered-input problem. This ordered-list promise is absent from selection under an unknown strict order.

Success as a function of query budget is itself an established question. The multiplicative adversary method gives success-dependent lower bounds~\cite{Spalek2008}, with extensions to state generation and small success probabilities~\cite{AmbainisEtAl2011}. Direct-product bounds give exponentially small success when too few queries are available for independent instances of functions on finite alphabets~\cite{LeeRoland2012}. Tight success bounds are known for one-marked-item search up to Grover search's maximum~\cite{Zalka1999}. Success at each query count can also be optimized for coset identification with finite-group oracle representations~\cite{CopelandPommersheim2021}. A comparison oracle has relabelling symmetry, but this alone does not give it the group-representation structure required by that result.

General query methods provide a basis for analyzing these different oracle models. Semidefinite feasibility characterizes standard Boolean-query algorithms~\cite{BarnumSaksSzegedy2003} and finite families of unitary queries~\cite{barnum}. Applications include numerical optimization of small Boolean functions, with exact recovery on at most three bits~\cite{MontanaroJozsaMitchison2015}, and symmetric SDPs for exact translation-invariant ordered search~\cite{ChildsLandahlParrilo2007}. Negative-weight adversaries give query lower bounds at prescribed error~\cite{HoyerLeeSpalek2007}. State-conversion complexity at bounded error is characterized up to constants~\cite{LeeEtAl2011}. These complexity results address query cost. The general semidefinite frameworks provide a starting point, but a comparison-specific analysis is needed to characterize attainable rank-selection probabilities and obtain sharp bounds under a two-query budget.

In the comparison setting, this requires retaining both the strict-order relations and the constraints imposed by the sequence of physical operations. Comparison signs arise from a single strict order and cannot vary independently. In a coherent realization, the same input-independent operation acts between the two queries for every possible order. It preserves inner products, so the states for different orders cannot be chosen independently. A nonnegative rank distribution that sums to one and satisfies the polynomial degree bound need not be realizable by two queries. Quantum strategies and combs express constraints between successive operations~\cite{GutoskiWatrous2007,ChiribellaEtAl2009}. The resulting problem is to determine attainable rank-selection success under both the strict-order relations and these sequential constraints.

To address this problem, we study the largest attainable probability of selecting an object within a prescribed rank threshold. The input is an unknown strict order, and the budget is at most two noiseless coherent comparisons. We specialize the semidefinite description to two comparison queries and use it to obtain exact finite-size results and bounds as the number of objects varies. For minimum selection, we also bound the gap between an explicit algorithm family and the unrestricted optimum. The QOO paper announced the finite-size values and scaled-limit bounds considered below without proof~\cite{wuqoo}. We supply their definitions, proofs and certificates here. Our main contributions are as follows.
\begin{itemize}
\item \textbf{Complete model} (analytical).
For every size at least three, the model characterizes attainable success and recovers a finite ideal algorithm from each feasible point. For sizes at least seven, symmetry reduces the description to a size-independent number of coefficients, although the matrices depend on the size.

\item \textbf{Exact finite-size optima} (computer-assisted).
Matching attainable values and upper bounds certify exact optima for every rank threshold at seven objects and for minimum selection at eight.

\item \textbf{Explicit algorithm family} (computer-assisted).
For minimum selection, we construct an explicit real pair-register family for every size at least seven. We determine its exact minimum-selection optimum and give parameters that attain it.

\item \textbf{Uniform bounds and scaled limit} (computer-assisted).
A separate certificate from the complete model bounds all admissible two-comparison algorithms, so these results quantify the family's distance from unrestricted optimal performance. Together with a subset-transfer argument, the bounds prove that the unrestricted minimum-selection probability, multiplied by population size, converges to a constant near $7.44$. The family's asymptotic success is within $0.016\%$ of this unrestricted optimum. A rational feasible point nevertheless proves that the family's limiting scaled success is strictly smaller.
\end{itemize}

These results also quantify the improvement over classical selection. With two classical comparisons, selecting the best of three uniformly sampled objects is optimal~\cite{wuqoo}. For minimum selection, the ratio of optimal quantum to optimal classical success tends to approximately $2.48$. Two quantum comparisons outperform six classical comparisons exactly when there are at least $46$ objects.

The remainder of this paper is organized as follows. Section~\ref{sec:model} specifies the information interface and success criterion. Section~\ref{sec:feasibility} characterizes feasible two-comparison algorithms. Sections~\ref{sec:exact}--\ref{sec:uniform} give the finite-size optima, the explicit pair-register family and the unrestricted bounds. Section~\ref{sec:interpretation} compares quantum and classical selection, explains the scope of the comparison model and identifies remaining questions. Section~\ref{sec:conclusion} summarizes the main results. The appendices give the detailed certificate constructions and verification arguments.

\section{Model and success criterion}\label{sec:model}
The distinction between a comparison budget and an evaluation budget requires
an explicit interface. Let $[n]=\{1,\ldots,n\}$ label $n$ objects. An input is a
hidden strict order represented by distinct ranks $x_i\in[n]$. Rank one is
best. The algorithm receives no numerical rank or value information. Its
charged operation is the ordinary comparison-XOR oracle
\begin{equation}\label{eq:oracle}
 O_x\lvert i,j,b,\zeta_{\rm ws}\rangle
 =\lvert i,j,b\mathbin\oplus[x_i<x_j],\zeta_{\rm ws}\rangle,
 \qquad i\ne j.
\end{equation}
The address domain is $\{(i,j)\in[n]^2:i\ne j\}$, and $[x_i<x_j]$ is one when $x_i<x_j$ and zero otherwise. Here $b$ is a response bit and $\zeta_{\rm ws}$ is a finite workspace register. Free
operations are ideal and input-independent. The model permits arbitrary
finite workspace, private randomness, finite intermediate instruments, and a
complete final measurement with one output label. Each branch uses at most two
oracle calls. There is no postselection. The complete coefficient model holds for $n\ge3$. Unless an endpoint is
stated explicitly, we use $1\le r<n$. The evaluated certificates and
the uniform results start at seven objects. The signed orbital catalogue
of Appendix~\ref{app:symmetry} lists feature-pair overlap patterns with
their orientation signs and stabilizes at seven labels, so the same
symmetry reduction and uniform certificate apply from that size on.

For an algorithm $A$, let $R_A(x)$ be the absolute rank of its output, and
define
\begin{equation}\label{eq:value}
 \Q(n,r)=\sup_A\min_x\Pr\{R_A(x)\le r\}.
\end{equation}
In the terminology of the QOO formulation~\cite{wuqoo}, $\Q(n,r)$ is the optimal
alignment probability for a selected set of one object, the good-enough set
of the $r$ best objects, and alignment level one.
One fixed, uniformly random permutation of the labels converts any algorithm
with a given uniform-order average success into an algorithm with that success
on every strict order. The permutation is retained through both calls and
inverted at output. Hence the optimum uniform-order average equals
\eqref{eq:value}. This equivalence permits us to optimize an average and then recover a worst-order guarantee by relabelling.

The endpoint values are $\Q(n,0)=0$ and $\Q(n,n)=1$. In the same model, the
classical optimum with two comparisons is
\begin{equation}\label{eq:classical}
 C_2(n,r)=1-\frac{\binom{n-r}{3}}{\binom n3}.
\end{equation}
The minimum of a uniform three-object subset attains this value. For
completeness, a deterministic adaptive transcript partitions the queried
objects into components of the comparison graph. The component containing
the output has at most three objects. Among uniform hidden orders consistent
with that transcript, its absolute rank set is a uniform subset, independently
of its permitted internal order. Its selected order statistic therefore has
top-r probability no larger than that of a uniform triple minimum. Averaging
the transcripts and private randomness proves the upper bound. This is the
two-comparison specialization of the classical theorem in~\cite{wuqoo}. The same argument applies to any integer budget $q<n$. The output component has at most $q+1$ objects, whose absolute rank set is uniform after their internal comparisons are fixed. Its probability of containing the global minimum is at most $(q+1)/n$, which bounds the success of any selected object in it. A tournament on $q+1$ uniformly sampled objects attains this probability with $q$ comparisons. Therefore $C_q(n,1)=(q+1)/n$ for every $q<n$.

\subsection{Phase sectors and charged identity calls}\label{sec:sectors}
We now express the oracle in a form that permits coefficient calculations
without changing the resource unit. For an unordered edge $e=\{a,b\}$ with
$a<b$, define
\begin{equation}\label{eq:signs}
 d_e(x)=\begin{cases}+1,&x_a<x_b,\\-1,&x_a>x_b,\end{cases}
 \qquad d_0(x)=1.
\end{equation}
There are $1+\binom n2$ phase sectors. An ordinary call with address $(b,a)$
and response $\lvert-\rangle$ applies $d_e$. A response
$\lvert+\rangle$ on a fixed distinct pair implements $d_0$, and still costs
one oracle call. Address orientations and sector multiplicities can be
retained in workspace. Conversely, resolving any ordinary call into the
response $+/-$ basis gives these sectors, with known orientation signs
absorbed in free operations.

For early termination, dilate each finite instrument to an isometry and retain its outcome in an orthogonal workspace record. Control later operations on that record without intermediate measurement. Retain stopped outputs and use a charged response-plus call on a fixed distinct pair in each unused slot. A final measurement reproduces the original probabilities. Finite records and dilations therefore give one coherent two-slot protocol with only ordinary charged calls.

\section{Complete two-comparison feasibility}\label{sec:feasibility}
After two queries, the amplitudes have degree at most two in the comparison signs. The corresponding feature vectors satisfy relations on strict orders. We first remove this redundancy, then impose the constraints required to realize the coefficients by a quantum algorithm.

\subsection{Features and the order quotient}
Put $N=\binom n2$ and $\nu_{\triangle}=\binom n3$. The formal real feature spaces have
orthonormal coordinate bases
\begin{equation}\label{eq:features}
 \mathcal F_1=\operatorname{span}\{1,d_e\},\qquad
 \mathcal F_2=\operatorname{span}\{1,d_e,d_ed_f:e<f\}.
\end{equation}
Here $e_0$ is the coordinate vector of the constant feature, and $e_{(ab,ac)}$ is the coordinate vector of the quadratic feature $d_{ab}d_{ac}$. The other linear and quadratic coordinate vectors are defined in the same way. Let $\phi_1(x)$ and $\phi_2(x)$ be their evaluation vectors. Repeated edges
cancel because $d_e^2=1$. For each $a<b<c$, the strict-order identity
\begin{equation}\label{eq:triangle}
 d_{ab}d_{ac}-d_{ab}d_{bc}+d_{ac}d_{bc}=1
\end{equation}
gives the coefficient relation
$\kappa_{abc}=e_{(ab,ac)}-e_{(ab,bc)}+e_{(ac,bc)}-e_0$.

\begin{lemma}[Quadratic evaluation kernel]\label{lem:kernel}
For $n\ge3$, the linear features in $\mathcal F_1$ are independent as functions of strict orders. The complete evaluation kernel in $\mathcal F_2$ is the span of the
$\nu_{\triangle}$ independent vectors $\kappa_{abc}$. If $\mathsf{U}_{\triangle}$ contains these vectors as
columns and $\boldsymbol{v}_{\triangle}=\mathsf{U}_{\triangle}\ind$, the orthogonal projector onto the evaluated quotient is
\begin{equation}\label{eq:projector}
 \Pi_{\rm ord}=I-\frac13\mathsf{U}_{\triangle}\mathsf{U}_{\triangle}^T+\frac{\boldsymbol{v}_{\triangle}\boldsymbol{v}_{\triangle}^T}{3(\nu_{\triangle}+3)},\qquad
 \dim\im \Pi_{\rm ord}=1+N+\binom N2-\nu_{\triangle}.
\end{equation}
In particular, $\Pi_{\rm ord}\phi_2(x)=\phi_2(x)$ on every strict order.
\end{lemma}
\begin{proof}
For $n=3$, evaluation on the six strict orders has rank four in
$\mathcal F_1$ and rank six in the seven-dimensional $\mathcal F_2$,
whose kernel is the single triangle vector. For $n\ge4$, swapping labels adjacent in rank changes only their comparison sign. Subtraction therefore sets each
coefficient of a linear relation to zero, and then its constant is zero.
Reversing the whole order changes the sign of the linear part of a quadratic
relation and preserves its other parts, so subtraction removes that linear
part by the preceding argument. Now keep an edge's endpoints adjacent and
subtract before and after their swap. For a disjoint edge, place its endpoints
consecutively elsewhere and swap them too. This second subtraction isolates
its quadratic coefficient, which is zero. Only edge pairs in one triangle
remain. For $a<b<c$, write their coefficients as $a_{\triangle}$ for $d_{ab}d_{ac}$, $b_{\triangle}$ for $d_{ab}d_{bc}$ and $c_{\triangle}$ for $d_{ac}d_{bc}$. When $a,b$ are adjacent in rank,
$d_{ac}=d_{bc}$. Moving $c$ from immediately before that pair to immediately
after it in the subtracted relation gives $a_{\triangle}+b_{\triangle}=0$, independently of every
other label. Repeating the argument with $b,c$ adjacent in rank gives $b_{\triangle}+c_{\triangle}=0$.
Thus $(a_{\triangle},b_{\triangle},c_{\triangle})=a_{\triangle}(1,-1,1)$. Subtract these multiples of
\eqref{eq:triangle} for every triangle. All nonconstant terms vanish,
and evaluation sets the remaining constant to zero. Distinct triangles have disjoint quadratic
supports, so these relations are independent. Their Gram matrix is
$\mathsf{U}_{\triangle}^T\mathsf{U}_{\triangle}=3I+\ind\ind^T$. Inverting this matrix gives
\eqref{eq:projector}.
\end{proof}

Lemma~\ref{lem:kernel} can also be derived from the Varchenko--Gelfand presentation of the real functions on
the chambers of the braid arrangement~\cite{VarchenkoGelfand1987}, together with the associated graded
presentation in Proposition~2.6 of Yagi and Yoshinaga~\cite{YagiYoshinaga2025}. The chambers are the strict
orders, the generators are the comparison indicators $(1+d_{ab})/2$, and the relation for each triple of
labels corresponds to \eqref{eq:triangle}. In degree at most two, the square and triple relations span the
relation space.

The quotient $\mathcal R=\im\Pi_{\rm ord}$ evaluates injectively. We equip it with the inner product inherited from $\mathcal F_2$. In this metric, the projected constant has squared norm $\|\Pi_{\rm ord}e_0\|^2=3/(\nu_{\triangle}+3)$. Let $J_0$ include $\mathcal F_1$ in $\mathcal F_2$, and $J_e$ multiply by $d_e$ and cancel repeated edges. These coordinate isometries give
\begin{equation}\label{eq:maps}
 \mathcal B_\ell=\Pi_{\rm ord}J_\ell,
 \qquad \mathcal B_\ell^T\phi_2(x)=d_\ell(x)\phi_1(x).
\end{equation}
These maps describe the second query on the supported quotient.

\subsection{The two causal equalities}\label{sec:causal}
For an output label $i$, define the exact objective matrix
\begin{equation}\label{eq:objectives}
 F_i^{(r)}=\E_x\!\left[
 \ind_{\{x_i\le r\}}\phi_2(x)\phi_2(x)^T\right],
\end{equation}
where $x$ is uniform over all $n!$ orders. The next theorem specializes finite-query semidefinite feasibility and recovery~\cite{BarnumSaksSzegedy2003,barnum} to the comparison-sign quotient. Its causal equalities express successive normalization, as in strategies and combs~\cite{GutoskiWatrous2007,ChiribellaEtAl2009}.

\begin{theorem}[Complete coefficient model]\label{thm:model}
For $n\ge3$ and $0\le r\le n$, $\Q(n,r)$ is the value of the semidefinite
program
\begin{equation}\label{eq:primal}
\begin{aligned}
 \max\quad &\sum_{i=1}^n\langle F_i^{(r)},\mathcal Y_i\rangle,\\
 \text{subject to}\quad
 &w_\ell\ge0,\quad\sum_\ell w_\ell=1,\quad \mathcal X_\ell\succeq0,\quad \mathcal Y_i\succeq0,\\
 &\sum_\ell\mathcal X_\ell=\diag(w),\\
 &\sum_i\mathcal Y_i=\sum_\ell\mathcal B_\ell\mathcal X_\ell\mathcal B_\ell^T,\qquad \Pi_{\rm ord}\mathcal Y_i\Pi_{\rm ord}=\mathcal Y_i.
\end{aligned}
\end{equation}
The $\mathcal X_\ell$ act on $\mathcal F_1$, the $\mathcal Y_i$ act on $\mathcal F_2$, and
$\langle A,B\rangle=\tr(A^TB)$. Every feasible point has a finite ideal
ordinary-XOR realization with the prescribed output probabilities
$\phi_2(x)^T\mathcal Y_i\phi_2(x)$. Conversely, every protocol in
Section~\ref{sec:model} has a feasible point with the same uniform-order
objective. Real matrices suffice even when complex protocols are allowed.
\end{theorem}
\begin{proof}
Partition the prepared state by its first-query phase sector. Its coefficient
Gram matrix is $\diag(w)$. Immediately before the second query, partition
the state by the second-query sectors. Their coefficient Gram matrices
$\mathcal X_\ell$ sum to $\diag(w)$ because the intervening operation is an isometry.
After the second query, the sector coefficient maps are multiplied by $\mathcal B_\ell$.
Projection by $\Pi_{\rm ord}$ does not change any evaluated state. A complete final
measurement gives supported positive coefficient Grams $\mathcal Y_i$ whose sum is
the full postquery Gram matrix. This proves both equalities in
\eqref{eq:primal}. Taking real parts of Hermitian positive matrices
preserves positivity, the real equalities, and the objectives.

For the converse, factor $\mathcal X_\ell=\mathsf L_\ell^T\mathsf L_\ell$ and stack the factors. The first
equality gives equal Grams for the weighted first-query feature columns and
the stacked factors. Therefore one input-independent isometry maps the first
state family to the second. After the second query, stack $\mathsf L_\ell\mathcal B_\ell^T$ into
$\mathsf{C}_{\rm qry}$. Factor and stack the $\mathcal Y_i$ into a matrix $\mathsf{C}_{\rm out}$. The second equality
gives $\mathsf{C}_{\rm qry}^T\mathsf{C}_{\rm qry}=\mathsf{C}_{\rm out}^T\mathsf{C}_{\rm out}$, so $\mathsf{C}_{\rm qry}u\mapsto \mathsf{C}_{\rm out}u$ is a well-defined isometry on
$\im \mathsf{C}_{\rm qry}$. Extend both isometries to finite unitaries with finite ancillas.
Measure the orthogonal output-label summands and assign unused directions to
one label. The resulting measurement is complete. Section~\ref{sec:sectors}
implements each phase slot by one charged ordinary call. Relabelling converts
the average objective into the same success on every order. This specializes
the finite-query Gram recovery method of~\cite{BarnumSaksSzegedy2003,barnum}.
\end{proof}

The feasible set in \eqref{eq:primal} is bounded. Indeed,
$\sum_\ell\tr \mathcal X_\ell=1$, and $\|\mathcal B_\ell\|\le1$ bounds the trace of
$\sum_i\mathcal Y_i$. It is also closed. Thus an optimum exists for each finite
$n,r$, and the recovery argument gives a finite ideal algorithm that attains the optimum. 

A useful sufficient dual consists of a scalar objective $\theta$, a real
symmetric $H$ on $\mathcal F_1$ and a supported real symmetric $Z$
on $\mathcal R$, with
\begin{equation}\label{eq:dual}
 \theta \ge H_{\ell\ell},\qquad H-\mathcal B_\ell^TZ\mathcal B_\ell\succeq0,\qquad
 Z-F_i^{(r)}\succeq0\quad\text{on }\mathcal R.
\end{equation}
The identity
\begin{equation}\label{eq:gap}
\begin{aligned}
 \theta -\sum_i\langle F_i^{(r)},\mathcal Y_i\rangle
 ={}&\sum_\ell w_\ell(\theta -H_{\ell\ell})\\
 &+\sum_\ell\langle H-\mathcal B_\ell^TZ\mathcal B_\ell,\mathcal X_\ell\rangle
 +\sum_i\langle Z-F_i^{(r)},\mathcal Y_i\rangle
\end{aligned}
\end{equation}
proves weak duality directly. A feasible primal and dual with zero gap therefore determine the optimum.

The causal equalities are essential to this characterization. Nonnegativity, normalization and the rank-degree bound together permit rank laws that no two-comparison algorithm attains. The distinction between a polynomial degree bound and algorithmic realizability also appears for Boolean queries. Block-multilinear approximation need not determine bounded-error query complexity~\cite{AaronsonEtAl2016}, and the exact characterization by completely bounded forms imposes a norm constraint in addition to degree~\cite{ArunachalamEtAl2019}. Gribling and Laurent formulate the tensor completely bounded norm as a semidefinite program and obtain a corresponding characterization for Boolean query complexity~\cite{GriblingLaurent2019}.

For comparison ranks, the normalized polynomial $p(k)\propto(k-3.4)^2(k-6.5)^2$ on $k=1,\ldots,7$ is nonnegative and has degree four, as the degree bound of~\cite{wuqoo} permits, but $p(1)\approx0.761$ exceeds the seven-object optimum of Theorem~\ref{thm:exact}, about $0.7166$. It is therefore not the rank law of any two-comparison algorithm. The exact linear program over all nonnegative normalized degree-four laws on seven ranks has optimum $p(1)=45/56\approx0.80357$, also above the certified seven-object optimum of about $0.7166$.

\subsection{A fixed-size symmetry description}\label{sec:symmetry}
Relabelling reduces the growing feature matrices to four signed orbital algebras, whose bases index feature-pair orbits with orientation signs. Related coding bounds use Terwilliger algebras~\cite{Schrijver2005}. Group averaging preserves SDP feasibility and objective~\cite{gatermannparrilo}, and invariant algebras reduce positivity to smaller blocks~\cite{BachocEtAl2012}. Here Table~\ref{tab:algebras} lists the identity-query, fixed-edge, global quadratic and fixed-output spaces, with both query layers retained. If a relabelling fixes a feature pair but reverses its sign, the corresponding invariant matrix entry is zero. The remaining support patterns give these dimensions for every $n\ge7$.

\begin{table}[tb]
\centering
\caption{Signed coefficient spaces}\label{tab:algebras}
\begin{tabular}{lrr}\toprule
Algebra & Dimension & Self-adjoint dimension\\\midrule
$\operatorname{End}_{S_n}(\mathcal F_1)$ &3&3\\
$\operatorname{End}_{S_2\times S_{n-2}}(\mathcal F_1)$ &14&11\\
$\operatorname{End}_{S_n}(\mathcal F_2)$ &34&25\\
$\operatorname{End}_{S_{n-1}}(\mathcal F_2)$ &166&99\\\bottomrule
\end{tabular}
\end{table}

For signed orbital basis $A_j$, use the physical trace weights $D_j=\tr(A_j^TA_j)>0$. Left multiplication by a self-adjoint $S=\sum_j s_jA_j$, paired with this metric, preserves positivity~\cite{deklerk},
\begin{equation}\label{eq:regular}
 \mathcal K_A(s)_{pq}=\tr(A_p^TSA_q),\qquad
 S\succeq0\ \Longleftrightarrow\ \mathcal K_A(s)\succeq0.
\end{equation}
The forward implication evaluates $\tr(\mathsf M_{\rm test}^TS\mathsf M_{\rm test})$ for any algebra element $\mathsf M_{\rm test}$. Conversely, a negative spectral projector is a polynomial in $S$ in the same algebra and makes that form negative. The regular matrix is determined by the coefficients $s_j$.

The reduced flow equalities have three and 34 coordinates. Supported self-adjoint output and global spaces have dimensions 63 and 16. Appendix~\ref{app:symmetry} gives their weighted maps and top-r objectives. The coefficient count is fixed, but products, support, trace weights and positivity depend on $n$, so evaluating $\Q(n,r)$ remains a separate step.

\section{Exact finite-size optima}\label{sec:exact}
The complete model permits separate optimization for each rank threshold.
We now evaluate all thresholds at size seven and minimum selection at size
eight. The values are algebraic, and their definitions consist of rational
polynomial systems with isolating boxes. 

\begin{theorem}[Exact seven-object and eight-object values]\label{thm:exact}
In the model of Section~\ref{sec:model} (computer-assisted proof, Appendix~\ref{app:exact}),
\begin{equation}\label{eq:exactvalues}
 \Q(7,r)=\vartheta_r\quad(1\le r\le4),\qquad
 \Q(7,r)=1\quad(5\le r\le7),\qquad
 \Q(8,1)=\vartheta_8.
\end{equation}
The exact numbers $\vartheta_1,\vartheta_2,\vartheta_3,\vartheta_4,\vartheta_8$ are defined in
Appendix~\ref{app:exact}. Each has a matching primal and dual satisfying
\eqref{eq:primal}--\eqref{eq:dual}, and a finite ideal algorithm attains it
with two charged ordinary comparisons.
\end{theorem}

Table~\ref{tab:values} reports the rounded values and their classical
counterparts from \eqref{eq:classical}. Each nontrivial top-r entry has its
own certificate. 

\begin{table}[tb]
\centering
\caption{Optimal selection probabilities}\label{tab:values}
\begin{tabular}{ccrr}\toprule
Objects $n$ & Accepted ranks $r$ & Quantum $\Q(n,r)$ & Classical $C_2(n,r)$\\\midrule
7&1&0.716632459&$3/7$\\
7&2&0.922214198&$5/7$\\
7&3&0.980507441&$31/35$\\
7&4&0.996592752&$34/35$\\
7&5,6,7&1&1\\
8&1&0.658596828&$3/8$\\\bottomrule
\end{tabular}
\end{table}

\subsection{How the certificates establish optimality}
Each certificate pairs a feasible algorithm with a dual upper bound of the same value. It specifies affine dual pencils and an active face through an ordered rational polynomial system. It includes a rational center, a preconditioner and a closed box of radius $10^{-50}$. Exact interval contraction proves one nonsingular algebraic zero there. Its active blocks have the form
\begin{equation}\label{eq:kernelblocks}
 S_b=\begin{pmatrix}A_b&B_b\\B_b^T&C_b\end{pmatrix},\qquad
 A_bK_b+B_b=0,\qquad C_b+B_b^TK_b=0.
\end{equation}
Positive leading blocks throughout the box and the root kernel identities give
\begin{equation}\label{eq:singularpsd}
 S_b=\begin{pmatrix}I\\-K_b^T\end{pmatrix}
 A_b\begin{pmatrix}I&-K_b\end{pmatrix}\succeq0.
\end{equation}
Inactive blocks and supported multiplier matrices are positive throughout the same $10^{-50}$ box. Exact congruences establish positivity of the full dual matrices. We construct a primal from the positive kernel multipliers. The restored stationarity equations ensure that both causal equalities hold. Complementarity makes \eqref{eq:gap} zero. Appendix~\ref{app:exact} gives the root functionals and finite recovery. The isolated root selects one matching pair, without classifying all optimal algorithms or dual points.

\subsection{The attained rank law}
The rank law also shows how a minimum-optimal algorithm performs at other rank thresholds. For one such seven-object algorithm, the law is approximately
\begin{equation}\label{eq:ranklaw}
\begin{aligned}
 (&0.716632459,\ 0.182916624,\ 0.025273689,\ 0.015956247,\\
  &0.022465132,\ 0.007549421,\ 0.029206428).
\end{aligned}
\end{equation}
Its exact probabilities are nonnegative, normalized and have zero fifth and sixth differences, as the two-comparison degree bound requires~\cite{wuqoo}. The amplitude-degree argument is the polynomial method~\cite{BealsEtAl2001}. Rank seven is more probable than rank six, so the law is not monotone. Its top-2 success is approximately $0.899549083$, below $\vartheta_2\approx0.922214198$.

\section{A pair-register family for arbitrary size}\label{sec:family}
To obtain a minimum-selection guarantee for every $n\ge7$, we construct a family of two-comparison algorithms on a real pair register. Its free operations commute with relabelling and use no further workspace. Optimizing within this family gives a lower bound on the unrestricted optimum.

\subsection{Physical operations}
We represent a real pair-register state by a matrix $A\in\R^{n\times n}$ with the Frobenius norm. The initial state has one amplitude on diagonal pairs and another on off-diagonal pairs. The algorithm applies a comparison, a free operation $U$, a second comparison and a free operation $V$. We then measure the first label. Both free operations are orthogonal and commute with relabelling.

Under relabelling, $\R^{n\times n}=2(n)\oplus3(n-1,1)\oplus(n-2,2)\oplus(n-2,1,1)$ for $n\ge4$, with partitions denoting real irreducibles. The last two multiplicities are one, so the orthogonal commutant is $O(2)\times O(3)\times\{\pm1\}^2$ in orthonormal multiplicity coordinates. Both free operations $U,V$ lie in this group. The sixteen Cayley sign charts cover the commutant's success probabilities up to closure, as proved in Appendix~\ref{app:branches}.

Write $I$ and $J$ for the identity and all-ones matrices. The simultaneous row-column permutation action has the orthogonal decomposition
\begin{equation}\label{eq:decomp}
 A=aI+b(J-I)+\diag(z_0)+S(z_1)+T(z_2)+R_s+R_a,
\end{equation}
where each $z_j$ has zero sum, $S(z)_{ij}=z_i+z_j$ off the diagonal with
zero diagonal, and $T(z)_{ij}=z_i-z_j$. The symmetric and antisymmetric
residuals $R_s,R_a$ have zero diagonal and zero row and column sums.
The multiplicity metrics are
\begin{equation}\label{eq:metrics}
 G_2=\diag(1,n-1),\qquad G_3=\diag(1,2(n-2),2n).
\end{equation}
Thus the scalar norm is $n(a^2+(n-1)b^2)$ and the standard-component norm
is $\|z_0\|^2+2(n-2)\|z_1\|^2+2n\|z_2\|^2$.

For a real skew matrix $\Omega $ and positive diagonal $G$, define
\begin{equation}\label{eq:cayley}
 \Cay(G,\Omega )=(I-G^{-1}\Omega )^{-1}(I+G^{-1}\Omega ).
\end{equation}
The $G_2$ block has one upper skew entry, denoted $u_2$ or $v_2$ for $U$ or $V$. For the $G_3$ block, denote the upper skew entries by
$(u_{01},u_{02},u_{12})$ for $U$ and $(v_{01},v_{02},v_{12})$
for $V$, in the coordinate order $(z_0,z_1,z_2)$. An operation $\mathcal W_n$ applies the scalar Cayley block to $(a,b)$ and $\Cay(G_3,\Omega_3)\diag(1,1,\eta)$ to the standard components, with $\eta$ acting first. It preserves $R_s$ and multiplies $R_a$ by $\sigma$, where $\eta,\sigma\in\{\pm1\}$. These operations are orthogonal and commute with relabelling.

Define the effective phase query by
\begin{equation}\label{eq:pairphase}
 (\mathcal D_x A)_{ij}=
 \begin{cases}A_{ii},&i=j,\\
 \operatorname{sgn}(x_i-x_j)A_{ij},&i\ne j.
 \end{cases}
\end{equation}
For $e=\{i,j\}$ the off-diagonal phase is $-d_e$ for $i<j$ and $d_e$ for $i>j$. Retain both labels in workspace. Use response-minus at $(i,j)$ off the diagonal and response-plus on a fixed distinct pair on the diagonal. This isometric encoding extends to a finite unitary, and its inverse clears the ancillas. Thus \eqref{eq:pairphase} uses one ordinary call, also on superpositions of diagonal and off-diagonal pairs.

From a normalized state proportional to $aI+b(J-I)$, apply
\begin{equation}\label{eq:Wcircuit}
 V\mathcal D_xU\mathcal D_x,
\end{equation}
and measure the first label. The first query supplies comparison signs, $U$ redistributes amplitudes, and the second query and $V$ convert their interference into output probabilities. Define $\PW(n)$ as the supremum of the probability of selecting the best object over finite real $\mathcal W_n$ parameters and the sixteen tuples $(\eta_U,\sigma_U,\eta_V,\sigma_V)$, one per sign sheet. The real preparation amplitudes $a,b$ also vary in this supremum. Covariance makes success constant over orders, and $\PW(n)\le\Q(n,1)$. Appendix~\ref{app:branches} proves chart coverage and gives numerical circuit parameters.

After the first query and $U$, the scalar part is $\alpha I+\beta(J-I)$ in the construction of Appendix~\ref{app:branches}. The output coordinates use $c=\sigma_V\beta$. We call the branch with $\alpha>0$ the head and the branch with $\alpha=0$ the tail. Both retain the nonzero off-diagonal scalar component $\beta$. The branch
change therefore removes the diagonal scalar component. 

\subsection{Sharp family values}\label{sec:familyvalue}
\begin{theorem}[Sharp values in the pair-register family]\label{thm:family}
Write $t=1/n$ for inverse size. There are algebraic functions
$P_H(t)$ on $[1/24,1/7]$ and $P_T(t)$ on $[0,1/23]$, defined by the isolated rational systems (computer-assisted proof, Appendix~\ref{app:branches}). Then
\begin{equation}\label{eq:familyvalue}
 \PW(n)=\begin{cases}
 P_H(1/n)/n,&7\le n\le23,\\
 P_T(1/n)/n,&n\ge24.
 \end{cases}
\end{equation}
For each $n$, a finite circuit with sign tuple $(-1,+1,+1,+1)$ attains the selected value. The two selected branches meet at one transition $\tau$ with
\begin{equation}\label{eq:switch}
 \frac{434161}{10^7}<\tau<\frac{434162}{10^7},\qquad
 23<\tau^{-1}<24.
\end{equation}
The family constant $c_W=P_T(0)$ satisfies
\begin{equation}\label{eq:cW}
 7.43703431<c_W<7.43703432.
\end{equation}
For every integer $n\ge24$,
\begin{equation}\label{eq:familyerror}
 \frac{c_W}{n}-\frac{21}{n^2}<\PW(n)
 <\frac{c_W}{n}-\frac{18}{n^2}.
\end{equation}
\end{theorem}

\begin{proof}[Proof structure]
To prove optimality within the family, we combine a quadratic upper bound with an algorithm that attains it. The upper bound applies to every physical circuit in the declared family, for every allowed sign tuple. The fourteen coordinates are
\[
 (\alpha,A,L_0,\boldsymbol{Q}_0,s,c,g_0,g_1,g_2,b_{\rm out},L_1,\boldsymbol{Q}_1,L_2,\boldsymbol{Q}_2).
\]
They describe the scalar, standard and residual output amplitudes of Appendix~\ref{app:family} and satisfy $e_1=\cdots=e_5=0$ in \eqref{eq:five}. The quadratic identity \eqref{eq:globalidentity} bounds normalized success by $p$ when $np-m_2\ge0$, $D_m\succeq0$ and $\mathsf S_9\succeq0$. Solutions of two smaller stationary systems give algorithms that attain the bound. In the scaled coordinates of Appendix~\ref{app:branches}, $\pi=np$. On the tail, the slack $\gamma_T=\pi-m_{2,T}$ defined there is nonnegative up to the transition and selects the valid branch. The head sets $m_2=\pi$ and retains the positive diagonal scalar amplitude. Appendix~\ref{app:branches} proves these signs and the matching physical lift.

The systems in Appendix~\ref{app:branches} specify isolated nonsingular
roots on continuous parameter intervals. Their leading slack blocks are positive on the posterior root enclosures of radius below $10^{-9}$ in every parameter cell. Exact kernel identities
restore the singular full slacks at the selected roots and remain valid
at the branch transition. The resulting dual bound matches the value of
the physical lift. A single standard rotation maps both required input
columns to their output columns, and its inverse Cayley denominator is
nonzero. Thus finite parameters attain the bound. Finally, the certified
actual-root derivative satisfies $-21<P_T'(t)<-18$ on $[0,1/23]$.
Integrating between zero and $1/n$ proves \eqref{eq:familyerror}.
\end{proof}

The endpoint $t=0$ defines the algebraic limit of $n\PW(n)$. At size seven, $\PW(7)\approx0.711812116$ is strictly below $\Q(7,1)\approx0.716632459$, so at this size the pair-register restriction lowers the optimum.

\section{Uniform unrestricted bounds and the scaled limit}\label{sec:uniform}
The family gives a lower bound on the unrestricted optimum. We now derive an upper bound from the complete model. A subset-transfer argument then proves that the scaled optimum converges.

\begin{theorem}[Uniform bounds]\label{thm:uniform}
Set
\begin{equation}\label{eq:constants}
 c_U=\frac{743815019}{10^8},\qquad c_L=\frac{743703431}{10^8}.
\end{equation}
For every finite integer $n\ge7$ (computer-assisted proof, Appendix~\ref{app:uniform}),
\begin{equation}\label{eq:upper}
 \Q(n,1)\le\min\{1,c_U/n\}.
\end{equation}
For every integer $n\ge24$,
\begin{equation}\label{eq:gapfamily}
 \frac{c_L}{n}-\frac{21}{n^2}<\PW(n)\le\Q(n,1),\qquad
 0\le\Q(n,1)-\PW(n)
 <\frac{27897}{25000000n}+\frac{21}{n^2}.
\end{equation}
Moreover, $n\Q(n,1)$ is nondecreasing and bounded. Therefore
\begin{equation}\label{eq:limit}
 c_*:=\lim_{n\to\infty}n\Q(n,1)
 \quad\text{exists},\qquad c_W\le c_*\le c_U,
 \qquad \Q(n,1)=\frac{c_*}{n}+o(1/n).
\end{equation}
\end{theorem}

\subsection{A uniform upper certificate}
To obtain a bound valid for every $n\ge7$, we use fixed-size matrices to certify positivity of the complete-model dual slacks. A flag is a comparison-sign feature with specified shared labels and otherwise unmarked labels. Each relevant signed kernel has a shared-subset decomposition
\begin{equation}\label{eq:flagdecomp}
 S(n)=\sum_\mathcal{J} \mathsf T_\mathcal{J}\mathsf F_{|\mathcal{J}|}\mathsf T_\mathcal{J}^T.
\end{equation}
Here $\mathcal J$ is a literal shared-label subset, $\mathsf T_\mathcal J$ a signed selection matrix and $\mathsf F_k$ a fixed-size matrix. Positive $\mathsf F_k$ therefore give full-kernel positivity for every finite $n\ge7$. Appendix~\ref{app:uniform} proves the Boolean cancellation identity. Eleven rational positive blocks give complete query and output slacks with objective $c_U/n$, so \eqref{eq:gap} bounds all admitted protocols, also with complex amplitudes and arbitrary finite workspace. Flag positivity is sufficient for this upper bound.

\subsection{Subset transfer and its consequences}
To prove convergence of the scaled optimum, we transfer a finite two-comparison algorithm to a uniformly chosen subset of a larger instance. The construction preserves its comparison budget.

\begin{lemma}[Subset transfer]\label{lem:subset}
Suppose an algorithm on $k$ objects succeeds in minimum selection with
probability at least $p$ on every strict order and uses at most two
ordinary comparisons. For every $n\ge k$, an algorithm with the same
budget succeeds with probability at least $kp/n$.
\end{lemma}
\begin{proof}
Select a uniform $k$-element subset using private randomness and fix a
bijection from the local labels to that subset. Reversibly encode each
local query through the same bijection, make one global ordinary call,
and undo the encoding. Invert the local output through that bijection.
The global minimum belongs to the subset with probability $k/n$. For
every subset that contains it, the local minimum is the global minimum
and the success probability is at least $p$. Thus the total success is
at least $kp/n$. Sampling and retaining the subset uses no comparisons.
\end{proof}

\begin{proof}[Completion of Theorem~\ref{thm:uniform}]
The upper bound follows from the preceding certificate and the fact that a probability is at most
one. The lower bound from the family follows from Theorem~\ref{thm:family} and
$c_W>c_L$. Subtract this lower bound from $c_U/n$ and use
$c_U-c_L=27897/25000000$ to obtain \eqref{eq:gapfamily}.
For $n\ge k$, apply Lemma~\ref{lem:subset} to an
$\varepsilon$-optimal $k$-object protocol and let $\varepsilon$ decrease
to zero. It gives $n\Q(n,1)\ge k\Q(k,1)$. The constant $c_U$ bounds this
nondecreasing sequence. Its limit is at least $c_W$ because
$\PW(n)\le\Q(n,1)$ and $n\PW(n)\to c_W$. This proves
\eqref{eq:limit}.
\end{proof}

Monotonicity extends finite-size lower bounds on scaled success to larger populations. A certified feasible point therefore proves that the unrestricted limit exceeds the family constant. Certificates at two consecutive sizes also determine the threshold for outperforming six classical comparisons.

\begin{theorem}[The family is not asymptotically optimal]\label{thm:strict}
$c_*>7.43710307$, and in particular $c_*>c_W$ (computer-assisted proof, Appendix~\ref{app:strict}).
\end{theorem}
\begin{proof}
At $n_0=50000$, the exact rational point of Appendix~\ref{app:strict}
is feasible for \eqref{eq:primal} and has objective $v_0$ with
$n_0v_0>7.43710307>c_W$. Theorem~\ref{thm:uniform} gives $c_*\ge n_0\Q(n_0,1)\ge n_0v_0>7.43710307$, and in particular $c_*>c_W$.
\end{proof}

\begin{corollary}[The six-comparison threshold]\label{cor:threshold}
For $n\ge7$, two quantum comparisons have greater minimum-selection success than six classical comparisons if and only if $n\ge46$ (computer-assisted proof, Appendix~\ref{app:thresholdcert}).
\end{corollary}
\begin{proof}
The rational primal at size 46 in Appendix~\ref{app:thresholdcert} has
$46v>7$, while the complete-model dual at size 45 has $45\theta<7$.
Theorem~\ref{thm:uniform} makes $n\Q(n,1)$ nondecreasing, so the first
inequality holds at every larger size and the second rules out every size
from seven through 45. Six classical comparisons have optimum $7/n$ for
$n\ge7$, and select the best object with certainty when $n\le7$.
For $n\le6$, six classical comparisons select the best object with certainty, so two comparisons cannot exceed them. Thus the threshold statement holds for every $n$.
\end{proof}

The remaining uncertainty in $c_*$ is less than $0.00104712$. Theorems~\ref{thm:uniform} and~\ref{thm:strict} give $7.43710307<c_*\le7.43815019$, with an exact feasible-point lower bound above the printed lower endpoint. The argument does not decide whether the upper endpoint is attained.
The $n^{-2}$ remainder in \eqref{eq:familyerror} holds for the constructed
family, whereas the proved remainder for the unrestricted optimum in
\eqref{eq:limit} is $o(1/n)$.

The bounds leave the exact value of $c_*$ unresolved. Numerical solutions of the complete model and extrapolation motivate the conjecture that $c_*$ equals the flag optimum. This is the smallest upper constant that a dual certified by flag positivity can give, and its numerical value is about $7.43815017$. The extrapolation assumes an expansion in powers of $1/n$. Bounds on numerical error at the sampled sizes can be propagated to
the fitted intercept. These bounds do not cover extrapolation bias
or error in the numerical flag optimum. Appendix~\ref{app:numerics} gives the numerical evidence and its limitations.

\section{Interpretation and remaining questions}\label{sec:interpretation}
Matching certificates determine the stated unrestricted optima at seven and eight objects. At arbitrary size, Theorem~\ref{thm:uniform} bounds minimum selection, while the pair-register theorem is sharp within its family for every size at least seven. Classically, $q<n$ comparisons give optimum $(q+1)/n$ by Section~\ref{sec:model}. Corollary~\ref{cor:threshold} places the exact crossing of the six-comparison value $7/n$ at $n=46$. With two comparisons, the asymptotic success ratio is
\begin{equation}\label{eq:ratio}
 \frac{7.43710307}{3}<\frac{c_*}{3}\le\frac{7.43815019}{3}.
\end{equation}
For minimum selection, the equivalent blind-pick size is $n$ times the success probability. The equivalent blind-pick size~\cite{wuqoo} tends to $c_*\approx7.44$, against $3$ for two classical comparisons and $(3+\sqrt{13})/2\approx3.30$ for one quantum comparison. The ratio is approximately $2.48$, a constant-factor gain at a fixed comparison budget.

Figure~\ref{fig:scaled} compares size-scaled minimum success. Solid seven- and eight-object points are exact, the family curve uses certified branch enclosures, and hollow complete-model points are numerical. The one-comparison curve and limit use~\cite{wuqoo}. The second panel gives numerical $n(\Q-\PW)$ and $c_U-c_W$. The first panel's horizontal lines are the unrestricted upper constant and equivalent blind-pick sizes of two and six classical comparisons. All curves use exact comparisons and accept only the best object.
\begin{figure}[!htbp]
\centering
\includegraphics[width=.92\linewidth]{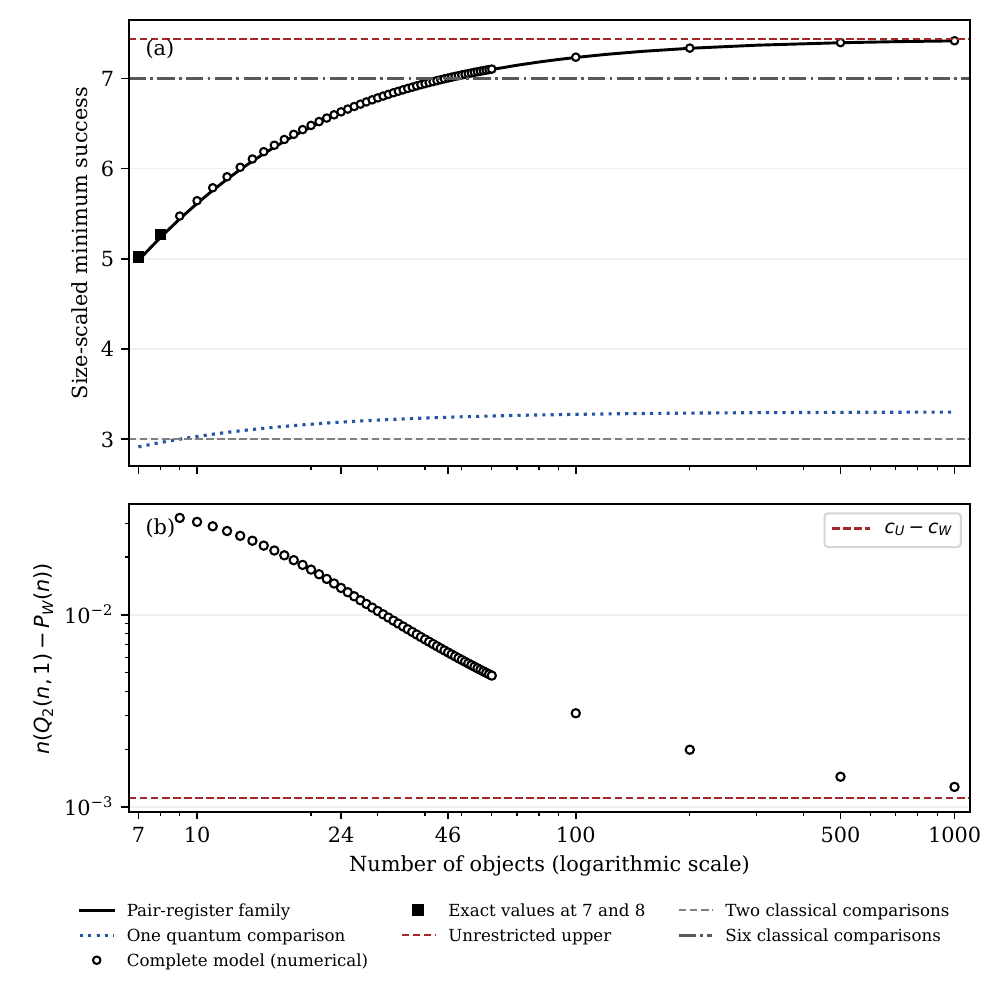}
\caption{Size-scaled probability of selecting the best object}\label{fig:scaled}
\end{figure}

The success probability depends on the information supplied by the oracle. If an oracle instead marks the best object, Grover search after $k$ iterations has success $\sin^2((2k+1)\theta)$, where $\sin\theta=n^{-1/2}$~\cite{BoyerEtAl1998}. Two iterations therefore find it with probability $\sin^2(5\theta)$.
This is about $25/n$ for large $n$, and is optimal in the range before
the first maximum~\cite{Zalka1999}. Comparison access reveals relative
order, and two comparisons reach about $7.44/n$. The numbers therefore
refer to different information interfaces.

In the comparison model, exact comparisons are charged and interquery unitaries are free. Gate complexity, running time and noisy access are separate questions. To build one comparison from an oracle that evaluates one object, computing and uncomputing both values costs four evaluation calls~\cite{wuqoo}. Two comparisons then cost eight evaluations. For $n\ge8$, eight evaluated objects contain the best one with probability $8/n$. Theorem~\ref{thm:uniform} gives $\Q(n,1)\le c_U/n<8/n$ for these sizes. Under this compute--uncompute implementation, the higher success per comparison does not give a gain per evaluation over sampling eight objects.

Two mathematical questions follow directly from the results. One is to
determine the unrestricted constant $c_*$ and the finite-size values
$\Q(n,1)$ beyond the evaluated cases. The values $\Q(n,r)$ for
$n=4,5,6$ are not certified exactly here. The numerical solutions of the unreduced programs appear in Appendix~\ref{app:numerics}. The other is to evaluate general
top-r objectives in the complete model, particularly when the accepted
fraction of the population is fixed. The exhibited minimum-optimal algorithm has top-2 success about $0.8995$, below $\vartheta_2\approx0.9222$, so it does not attain the top-2 optimum. Whether some minimum-optimal algorithm does remains open. 

\FloatBarrier
\section{Conclusion}\label{sec:conclusion}
We studied single-object ordinal selection under a branchwise budget of
two exact coherent comparisons. A complete coefficient-Gram model and its symmetry
reduction retain the causal and measurement constraints of the physical algorithms.
Matching algebraic certificates determine all top-r optima for seven objects.
They also determine the minimum-selection optimum for eight objects.
For every size at least seven, we determine the minimum-selection optimum within an explicit pair-register
family. A separate unrestricted certificate bounds its distance from the optimum.
A finite rational feasible point proves that the family is not asymptotically optimal.

The size-scaled unrestricted minimum-selection probability converges to a constant
in $(7.43710307,7.43815019]$.
Two quantum comparisons outperform six classical comparisons exactly when $n\ge46$.

\section*{Use of AI tools}
The author used AI tools to assist with literature research, mathematical
analysis, code development, and manuscript preparation. The author assumes
full responsibility for the content of the article.

\appendix
\section{Signed symmetry reduction}\label{app:symmetry}
The signed model certificate in the verification archive~\cite{WuTwoComparisonArchive} supplies the rational catalogue, products, maps and moments for the symmetry reduction of Section~\ref{sec:symmetry}. The following derivation gives their construction and the identities used to check them. In the first two subsections below, $p$ indexes an orbital and $k_p$ counts its free labels. In the third, $p$ is the coefficient vector of the projector.

\subsection{Finite support patterns}
A feature is empty, one oriented edge, or two distinct oriented edges. A permutation sorts each edge again and contributes its orientation sign. A feature pair acquires the product sign. Fix the anchors and stabilizer. If an element of the stabilizer reverses the sign, the entry is zero. Otherwise, normalize the signed orbital to $+1$ at a representative. Orbital supports are disjoint.

To enumerate the catalogue, classify endpoint equality patterns with anchors, endpoint permutations and orientations, reject loops and repeated edges, and identify signed isomorphisms. All nonzero types use at most seven labels. For two degree-one features this follows from their support and anchors, which use at most six. A global quadratic pair on eight labels consists of disjoint two-edge features. Swapping one edge's endpoints fixes the pair and reverses its sign. With an output anchor and at least eight labels, one feature has two disjoint edges, because two features with at most three vertices each and the anchor cover at most seven labels. At least three of its four vertices lie outside the other feature and the anchor, so one of its edges has both endpoints outside that set, and swapping them again reverses just one sign. Thus the one additional global and six additional anchored types vanish, and no new nonzero orbital appears above seven.

For an orbital $p$, let $a$ be the number of fixed anchors, $k_p$ its
number of free vertices, and $\nu _p$ its positive automorphism count. Let
$s=1$ for the unordered-edge stabilizer and $s=0$ otherwise. Then
\begin{equation}\label{eq:metricorb}
 D_p(n)=\tr(A_p^TA_p)=\frac{2^s(n-a)_{k_p}}{\nu _p},
 \qquad (z)_k=z(z-1)\cdots(z-k+1).
\end{equation}
Every factor is positive for $n\ge7$ because $a+k_p\le7$ for a
nonzero orbital. The signed transpose is
$A_p^T=\epsilon_pA_{\tau(p)}$, so self-adjoint coefficients obey
$s_{\tau(p)}=\epsilon_ps_p$. 

For multiplication, fix a representative of the output orbital $p$ with $v_p$ labels and sum over a middle feature. It introduces at most four new labels. For exactly $\ell$ new labels, count signed contributions on a canonical fresh $\ell$-set and multiply by $\binom{n-v_p}{\ell}$. The middle feature's two orientation signs cancel, so
\begin{equation}\label{eq:orbproduct}
 (A_iA_j)_p=\sum_{\ell=0}^{4}c_{pij\ell}
                  \binom{n-v_p}{\ell},
 \qquad c_{pij\ell}\in\mathbb Z.
\end{equation}
For every nonzero output orbital, $v_p\le7$. Thus \eqref{eq:orbproduct} applies at $n=7$, and direct enumeration at $n=7,\ldots,11$ identifies its degree-at-most-four polynomial. The checker also enumerates $n=12$ and tests associativity at $n=50000$. Direct maps and local moments depend only on abstract patterns. Their adjoints use the trace counts in \eqref{eq:metricorb}. These facts extend the reduction to every $n\ge7$, with binomial coefficients set to zero when too few labels are available.

\subsection{Independent representation counts}
For the label permutation representation $V_n$, $\mathcal F_1=\mathbf1\oplus\bigwedge^2V_n$ has global constituents $(n),(n-1,1),(n-2,1,1)$ once each. On $S_2\times S_{n-2}$ its multiplicities are $2,2,2,1,1$ for $(+,S^{(n-2)})$, $(-,S^{(n-2)})$, $(+,S^{(n-3,1)})$, $(-,S^{(n-3,1)})$, $(+,S^{(n-4,1,1)})$. Here $+$ and $-$ denote the trivial and sign representations of $S_2$, and $S^\lambda$ a real irreducible.

Quadratic edge pairs have three or four labels. Their representations
induce $S^{(3)}\oplus S^{(2,1)}$ from $S_3\times S_{n-3}$ and
$S^{(2,2)}\oplus S^{(1,1,1,1)}$ from $S_4\times S_{n-4}$,
with a trivial representation on the remaining labels.
Adding constant and linear features and applying the horizontal-strip
rule gives Table~\ref{tab:characters}. Restriction to one fixed label
removes a corner of each partition. The triangle kernel is the
three-subset permutation module with constituents
$(n),(n-1,1),(n-2,2),(n-3,3)$, each once. Subtracting it gives the
supported space $\mathcal R$. In the first two table columns $d=n$,
and in the last two $d=n-1$.
\begin{table}[tb]
\centering\small
\caption{Quadratic multiplicities}\label{tab:characters}
\begin{tabular}{lrrrr}\toprule
Partition&Global $\mathcal F_2$&Global $\mathcal R$&Fixed-label $\mathcal F_2$&Fixed-label $\mathcal R$\\\midrule
$(d)$&2&1&5&3\\
$(d-1,1)$&3&2&8&6\\
$(d-2,2)$&3&2&6&4\\
$(d-2,1,1)$&2&2&5&5\\
$(d-3,3)$&1&0&1&0\\
$(d-3,2,1)$&2&2&3&3\\
$(d-3,1,1,1)$&1&1&2&2\\
$(d-4,2,2)$&1&1&1&1\\
$(d-4,1,1,1,1)$&1&1&1&1\\\bottomrule
\end{tabular}
\end{table}
For multiplicities $m_\lambda$, the commutant and self-adjoint dimensions
are $\sum m_\lambda^2$ and $\sum m_\lambda(m_\lambda+1)/2$.
The first-layer counts are $(3,3)$ and $(14,11)$.
Formal quadratic counts $(34,25)$ and $(166,99)$ become $(20,16)$ and
$(101,63)$ on the support. Thus the supported self-adjoint dimensions 16
and 63 differ from the formal 25 and 99 in Table~\ref{tab:algebras}.
The archived character program independently recomputes them from signed
traces and exact character inner products.

\subsection{Physical pairing and the full reduced flows}
Denote the four algebras in Table~\ref{tab:algebras} by
$\mathcal A_{10},\mathcal A_{1e},\mathcal A_{20},\mathcal A_{2o}$,
in that order. Their diagonal trace metrics are
$D_{10},D_{1e},D_{20},D_{2o}$. The physical pairing of coefficient
vectors $s,x$ in the same algebra is $s^TDx$. Let $E_1$ and $E_2$
express a global operator in the edge-stabilizer and output-stabilizer
bases, respectively. Define the contractions
$C_0(S)=J_0^TSJ_0,\qquad C_e(S)=J_{\{1,2\}}^TSJ_{\{1,2\}}.$
These definitions determine all direct coefficient maps by finite
signed entries. Their summed weighted adjoints are
\begin{equation}\label{eq:adjoints}
\begin{aligned}
 T_1&=ND_{10}^{-1}E_1^TD_{1e},&
 T_0&=D_{20}^{-1}C_0^TD_{10},\\
 T_e&=ND_{20}^{-1}C_e^TD_{1e},&
 T_o&=nD_{20}^{-1}E_2^TD_{2o}.
\end{aligned}
\end{equation}
The factors $N=\binom n2$ and $n$ are the numbers of edge and output
cosets. They cannot be removed by rescaling the coefficient pairing.

Let $p$ be the global coefficient vector of the projector in
\eqref{eq:projector}, and $p_o=E_2p$. Use $x_0,x_e,y$ for the identity
query, edge-query, and output seed coefficients. With $*$ denoting
the product \eqref{eq:orbproduct}, the full reduced constraints are
\begin{equation}\label{eq:reducedflows}
\begin{aligned}
 x_0+T_1x_e&=(w_0,w_e,0),& w_0+Nw_e&=1,\\
 T_oy&=p*(T_0x_0+T_ex_e)*p,&p_o*y*p_o&=y.
\end{aligned}
\end{equation}
These equalities retain all three and all 34 flow coordinates. Nonnegative weights, signed self-adjointness and \eqref{eq:regular} impose positivity. Coset conjugation reconstructs \eqref{eq:primal}, and group averaging gives these seeds conversely, so the reduction preserves the optimized average.

For an output-anchored orbital $a$, let $s_a\le7$ be the size of its
support, with the output anchor included. Cancel doubled edges in the
product of its two features. Let $C_{a,\ell}$ be the signed sum over
local strict orders in which the output occupies position $\ell+1$.
The objective coefficient is
\begin{equation}\label{eq:momentformula}
 f_a(n,r)=\frac1{(n)_{s_a}}
 \sum_{\ell=0}^{s_a-1}C_{a,\ell}
 \sum_{j=\ell+1}^{s_a}\binom rj\binom{n-r}{s_a-j}.
\end{equation}
Indeed, the anchor is accepted exactly when at least $\ell+1$ of the uniformly assigned absolute ranks lie in $[r]$, and the relative order fixes the sign. The reduced objective is $nf(n,r)^TD_{2o}y$. At $r=1$, $nf_a(n,1)=C_{a,0}/(s_a-1)!$ is size-independent, as used in Appendix~\ref{app:uniform}.

\section{Exact certificates and finite recovery}\label{app:exact}
The rational systems in Table~\ref{tab:rootfiles} define the exact optima. Certificate indices are zero-based, whereas displayed output ranks are one-based. Here $\xi$ is the ordered root, $y$ its face coordinates, $G$ the evaluation metric, $\mathsf N$ a face basis, $\Lambda_b$ a kernel multiplier and $\eta$ the isolation residual. The five minimal-polynomial degrees have not been computed.

\begin{table}[tb]
\centering\small
\caption{Algebraic definitions}\label{tab:rootfiles}
\begin{tabular}{llrr}\toprule
Value & Certificate & Variables & Center denominator\\\midrule
$\vartheta_1$ & Seven-object minimum &38&$10^{80}$\\
$\vartheta_2$ & Seven-object top two &33&$10^{80}$\\
$\vartheta_3$ & Seven-object top three &34&$10^{80}$\\
$\vartheta_4$ & Seven-object top four &34&$10^{80}$\\
$\vartheta_8$ & Eight-object minimum &38&$10^{90}$\\\bottomrule
\end{tabular}
\end{table}

The verification archive specifies the rational polynomial systems, coordinate maps, face bases and isolating boxes for the seven-object minimum, the seven-object top-r values with $r=2,3,4$, and the eight-object minimum. Every box radius is $10^{-50}$. Integer polynomials, centers and preconditioners use the positive common denominators specified in each certificate. The certificate index gives the corresponding exact inputs and checking programs.

\subsection{Seven-object minimum}
Over all 5040 strict orders, $\mathsf V_1$ has the constant and 21 sign columns. Removing one quadratic feature per triangle gives 197 independent columns in $\mathsf V_2$, with exact maps
\begin{equation}
 D_\ell\mathsf V_1=\mathsf V_2\mathsf T_\ell,\qquad
 F_i^{(r)}=\frac1{5040}\mathsf V_2^T\Pi_i^{(r)}\mathsf V_2,
\end{equation}
where $D_\ell$ multiplies by $d_\ell$ and $\Pi_i^{(r)}$ indicates accepted orders. This basis represents the supported quotient with positive metric $G=\mathsf V_2^T\mathsf V_2$, so \eqref{eq:primal} and \eqref{eq:gap} apply.

The dual coordinates are $\boldsymbol{\theta}=(\theta ,\gamma,z_0,\ldots,z_{15})$, with
$H=\theta I+\gamma K_1$ and $Z=\sum_\ell z_\ell Z_\ell$. Here $K_1$
is the signed incidence form on edges with zero constant row, and the
$Z_\ell$ are a complete basis of invariant quadratic forms. Their explicit integer matrices and exact block congruences are supplied with the certificate.
For each identity-query, edge-query or output slack $S$, let $N_{\rm blk}$ be its number of full blocks. The supplied invertible matrix $\mathsf{V}_{\rm blk}$ satisfies the exact rational identity
\begin{equation}\label{eq:resolution7}
 \mathsf{V}_{\rm blk}^TS\mathsf{V}_{\rm blk}=\diag(\mathsf K_1,\ldots,\mathsf K_{N_{\rm blk}}),\qquad
 S=\sum_{b=1}^{N_{\rm blk}}W_b^T\mathsf K_bW_b,
\end{equation}
where $W_b$ is the corresponding row block of $\mathsf{V}_{\rm blk}^{-1}$. There are
16 blocks for each query operator and 61 output blocks. Every full-block
resolution weight in \eqref{eq:resolution7} is one. The certificate
identifies repeated blocks by exact congruences
$\mathsf K_b=s_bQ_b^T\mathsf K_{\mathrm{rep}(b)}Q_b$ with $s_b>0$.
Table~\ref{tab:sevenblocks} gives the full ranks. The seven-object boundary certificate specifies the eighteen representative blocks and their multiplicities. The exact checker reconstructs their individual ranks for each rank threshold.
\begin{table}[tb]
\centering\small
\caption{Seven-object slack ranks}\label{tab:sevenblocks}
\begin{tabular}{lrrr}\toprule
Value&Identity-query rank&Edge-query rank&Output rank\\\midrule
$\vartheta_1$&1&11&167\\
$\vartheta_2$&1&11&172\\
$\vartheta_3$&1&15&176\\
$\vartheta_4$&1&15&176\\\bottomrule
\end{tabular}
\end{table}
Put $\boldsymbol{\theta}=\mathsf Ny$, where the seven-object boundary certificate specifies the rational $18$ by $13$ face basis and its two gauges. This face makes six query pencils zero and fixes a one-dimensional kernel in another.
The remaining coordinates, kernel entries, and positive kernel
multipliers form the 38-variable system. If $\xi^{(1)}$ is its unique
root in the supplied box and $y$ is restored using the listed coordinate
map and gauges, the exact value is
\begin{equation}\label{eq:t1definition}
 \vartheta_1=\mathsf{N}_{0,:}y.
\end{equation}
This choice of face and gauge gives a matching primal--dual pair.

Six active pencils have size-rank pairs
$(2,1),(2,1),(3,2),(6,4),(4,3),(5,4)$. Writing them as in
\eqref{eq:kernelblocks} gives 27 kernel equations. Let
$Y_b=[K_b;I]$ and introduce a positive symmetric kernel multiplier
$\Lambda_b$. The stationarity polynomials are
\begin{equation}\label{eq:stationarity7}
 g_k=\mathsf{N}_{0k}-\sum_b
 \left\langle\frac{\partial S_b}{\partial y_k},
 Y_b\Lambda_bY_b^T\right\rangle.
\end{equation}
The gauge dependencies are $g_8-2g_9+g_{10}=0$ and $-4g_7-2g_8+4g_9+g_{12}=0$. Eleven stationarity equations and 27 kernel equations give the square system of degree at most three. Symmetry of $Y_b^TS_bY_b$ supplies the omitted lower entry of a two-column kernel equation.

The zero query blocks and persistent kernel add thirteen symmetric multiplier entries. Their $18$ by $13$ normal-column matrix is $C_\perp$. Fix eight entries $\mu_F$ and put $\boldsymbol r_F=\boldsymbol r-C_F\mu_F$, where $\boldsymbol r$ is the full objective gradient minus active contributions. Five independent columns $C_I$ and selected rows $E_J$ give
\begin{equation}\label{eq:stationarityrestore}
\begin{aligned}
 \mu_I&=(E_JC_I)^{-1}E_J\boldsymbol{r}_F,\\
 \boldsymbol{r}-C_\perp\mu
   &=\{I-C_I(E_JC_I)^{-1}E_J\}\boldsymbol{r}_F=\Pi g,\\
 \Pi&=\{I-C_I(E_JC_I)^{-1}E_J\}\mathsf{N}(\mathsf{N}^T\mathsf{N})^{-1},
 \qquad g=\mathsf{N}^T\boldsymbol{r}_F.
\end{aligned}
\end{equation}
The rank-five normal span is $\ker\mathsf N^T$, and $\mathsf N^TC_F=0$, so \eqref{eq:stationarityrestore} restores all eighteen coordinates when $g=0$. The other three seven-object certificates use the same construction. Their added kernels are positive on their supports. Undoing congruences and averaging over permutations gives positive primal matrices.

Normalization and the three-dimensional invariant first layer give the first flow. The second-flow residual $E$ is contravariantly invariant and orthogonal to all sixteen invariant quadratic forms. Since $GEG$ is covariantly invariant and
$\langle E,GEG\rangle
 =\tr\bigl((G^{1/2}EG^{1/2})^2\bigr)>0
 \quad\text{for }E\ne0,$
completeness gives $E=0$. Exact complementarity makes the gap zero. The first-query weights are strictly positive. The seven-object checker reconstructs these weights from the certificate seeds and verifies both complete flows. The same primal seed and exact-rank objectives give \eqref{eq:ranklaw}.

\subsection{Seven-object top-r values}
For $r=2,3,4$, let $\xi^{(r)}$ be the zero specified in
Table~\ref{tab:rootfiles}. Their exact value functionals are
\begin{align}
 \vartheta_2={}&\frac{13169719559234601}{723500000000000}
 +\frac1{1447}\bigl(-547290\xi_2+365310\xi_3+1003590\xi_4
 \nonumber\\
 &\quad+2013660\xi_5+525690\xi_6-382320\xi_7-262845\xi_8
       +12598\xi_9-1003590\xi_{10}\bigr),\label{eq:t2definition}\\
 \vartheta_r={}&4779\xi_1^{(r)}+999\xi_2^{(r)}+7776\xi_3^{(r)}
       +10692\xi_5^{(r)}+\xi_9^{(r)},\quad r=3,4.
       \label{eq:t34definition}
\end{align}
The unadorned $\xi$ in \eqref{eq:t2definition} denotes $\xi^{(2)}$.
The equal functional form for $r=3,4$ is evaluated at different roots.
All objective matrices use the accepted-rank indicator for the stated
$r$, with exactly $720r$ accepted strict orders per output label.

The $r=2$ system has 22 kernel equations and eleven independent
stationarity equations. For $r=3,4$, the systems have 23 kernel
equations, nine retained stationarity equations, and the two factor
equations
\begin{equation}\label{eq:factorrepair}
 u+15wv=0,\qquad
 \lambda_7-\frac{1134}{625}\lambda_{13}v^2=0.
\end{equation}
These factors follow from the three omitted variations of the full
Lagrangian, which is differentiated before the coordinate is fixed.
Write $w$ for the first kernel coordinate of block 7, and $z$ for the
unpermuted kernel column of block 13. Put $u=3z_1-z_4$ and
$v=z_2+z_3$. The raw pencils satisfy the exact variation identity
\begin{equation}\label{eq:variationrepair}
 \delta S_{13}=-\frac{4032}{5}L^T(\delta S_7)L,\qquad
 L=\begin{pmatrix}0&3&0&0&-1\\0&0&-15&-15&0\end{pmatrix}.
\end{equation}
The three matrices $\delta S_7$ span the symmetric $2$ by $2$
matrices. Let $y=(w,1)^T$, $z_*=(u,-15v)^T$, and let
$s_7,s_{13}>0$ be the chart scales. The derivative in each omitted
direction is $-\langle\delta S_7,M\rangle$, where
\begin{equation}
 M=s_7\lambda_7yy^T-\frac{4032}{5}s_{13}\lambda_{13}z_*z_*^T.
\end{equation}
With $f_1=u+15wv$ and
$f_2=\lambda_7-(1134/625)\lambda_{13}v^2$, exact expansion gives
\begin{equation}\label{eq:factororigin}
 M=s_7f_2yy^T-\frac{4032}{5}s_{13}\lambda_{13}f_1
   \begin{pmatrix}f_1-30wv&-15v\\-15v&0\end{pmatrix}.
\end{equation}
Thus the two factors in \eqref{eq:factorrepair} set all three omitted
derivatives to zero. The coefficient $1134/625$ comes from
$225(4032/5)s_{13}/s_7$ with the supplied chart scales. Genuine gauge dependencies restore the remaining face
directions. The positive normal lift then restores all eighteen dual
coordinates. 

Positivity of the required leading blocks, inactive blocks, and supported multiplier matrices on the whole root box of radius $10^{-50}$ is established by exact directed interval elimination.
The singular active slacks use \eqref{eq:singularpsd} at the root.
Seed contraction, covariance, and complete stationarity give both
causal equalities by the same invariant-metric argument used for
$\vartheta_1$. Each of the three pairs has zero gap separately. The complete coefficient identities and the factor equation maps are supplied in the chart files.

For $r=5,6,7$, a uniform triple minimum always has accepted rank,
because at most $n-r\le2$ objects are unaccepted, so every three-object subset contains an accepted object.
The two-comparison tree has rank law
$(15,10,6,3,1,0,0)/35$. Its success is one, which matches the trivial
probability upper bound. 

\subsection{Eight-object minimum}
For eight objects, the formal quadratic space has dimension 407 and
the supported quotient has dimension 351. There are 56 independent
triangle relations. Formula \eqref{eq:projector} becomes
\begin{equation}
 \Pi_{\rm ord}=I-\frac13\mathsf{U}_{\triangle}\mathsf{U}_{\triangle}^T+\frac1{177}\boldsymbol{v}_{\triangle}\boldsymbol{v}_{\triangle}^T,
 \qquad (\Pi_{\rm ord})_{00}=\frac3{59}.
\end{equation}
The exact value is
\begin{equation}\label{eq:t8definition}
 \vartheta_8=\frac12\xi_2^{(8)}+\frac76\xi_3^{(8)}+\frac43\xi_4^{(8)},
\end{equation}
where $\xi^{(8)}$ is the isolated 38-coordinate zero specified by the
eight-object certificate in the verification archive. That certificate fixes
the common polynomial denominator. The exact checker verifies the strict
48-digit rational enclosure.

The original dual has eighteen coordinates, and the selected zero
blocks impose five independent linear conditions. The eight-object boundary certificate specifies the resulting $18$ by $13$ face basis, gauges and ordered coordinates.
The same six size-rank pairs as for $\vartheta_1$ give 27 kernel equations.
Eleven full-Lagrangian stationarity equations complete the system.

The output blocks are compressions to projected Young images of
dimensions $(3,6,4,5,3,2,1,1)$, with positive resolution weights
$(1,6,14,15,35,20,21,15)$. The weighted dimensions sum to 351.
The supplied integer embeddings $\mathsf{E}_b$ have positive Grams $G_b$ and
satisfy an exact resolution
\begin{equation}\label{eq:resolution}
 \sum_b d_b\mathcal R_\mathcal{H}(\mathsf{E}_bG_b^{-1}\mathsf{E}_b^T)=\Pi_{\rm ord},
\end{equation}
where $\mathcal{H}$ fixes the output label and $\mathcal R_\mathcal{H}$ is its finite
group average. Each full projected image is invariant under a
supported self-adjoint operator commuting with $\mathcal{H}$. It is therefore
reducing, and \eqref{eq:resolution} reconstructs that operator from
its compressed blocks with positive weights. This proves that the
block positivity tests imply full supported positivity. Query blocks
have the analogous resolution of the identity.

The normal lift restores all eighteen stationarity coordinates, and
the reconstructed primal has weights
$w_0=34852813/10^8$ and $w_e=65147187/(28\cdot10^8)$,
where $w_0$ and the two gauges are local free choices inside the optimal face.
All three coordinates of the first-flow residual and all 34 coordinates of the second-flow residual are zero.
Exact complementarity establishes equality in \eqref{eq:gap}.
Consequently \eqref{eq:t8definition} is both an upper bound for every
admitted protocol and an attainable value.

\subsection{Numerical enclosures}
For $\vartheta_2$, coordinates $\xi_0,\ldots,\xi_{10}$ are $y_0,y_1,y_2,y_3,y_4,y_5,y_6,y_9,y_{10},y_{11},y_{12}$, in that order. For $\vartheta_3$ and $\vartheta_4$, they are $y_1,y_2,y_3,y_4,y_5,y_6,y_7,y_{10},y_{11},y_{12},y_{13}$ after the first coordinate is fixed and removed. The remaining coordinates are the kernel and symmetric multiplier entries in the ordered active-block list of each chart. For $\vartheta_8$, the eleven free face coordinates precede the nineteen kernel entries and eight multiplier entries specified by the eight-object boundary certificate. These dictionaries specify the coordinates in \eqref{eq:t2definition}, \eqref{eq:t34definition} and \eqref{eq:t8definition}.

Table~\ref{tab:longenclosures} gives outward decimal endpoints of strict
rational enclosures. The first four come from the supplied root boxes.
For $c_W$, a refinement of the tail system at $t=0$ gives a box of
radius $10^{-50}$ inside the existing nominal endpoint tube. Its exact
contraction bound is below $3\cdot10^{-49}$, so it isolates the same
root. The endpoint certificate in the verification archive supplies the extra
rational center and preconditioner. Its exact checker verifies these bounds.
\begin{table}[tb]
\centering\scriptsize
\caption{Strict algebraic enclosures}\label{tab:longenclosures}
\begin{tabular}{lrr}\toprule
Value&Lower endpoint&Upper endpoint\\\midrule
$\vartheta_1$&0.7166324589220098567409774509252509725695&0.7166324589220098567409774509252509725696\\
$\vartheta_2$&0.9222141983202458865353864960131301315266&0.9222141983202458865353864960131301315267\\
$\vartheta_3$&0.9805074412250647003163731992802528175364&0.9805074412250647003163731992802528175365\\
$\vartheta_4$&0.9965927518837980531425613637828416849287&0.9965927518837980531425613637828416849288\\
$c_W$&7.4370343126460654630868806012845718553825&7.4370343126460654630868806012845718553826\\
\bottomrule\end{tabular}
\end{table}

\subsection{Root isolation and physical realization}
For rational $f$, center $c$, preconditioner $R$ and $\rho_{\rm box}=10^{-50}$, a Krawczyk-type interval-Newton argument~\cite{rump} bounds
$\eta=\|Rf(c)\|_\infty,\qquad
 \kappa=\sup_{\|\xi -c\|_\infty\le \rho_{\rm box}}
          \|I-Rf'(\xi )\|_\infty.$
Exact rational bounds give $\kappa<1$ and $\eta+\kappa\rho_{\rm box}<\rho_{\rm box}$. Therefore $\xi\mapsto\xi-Rf(\xi)$ contracts the closed box strictly into itself. The derivative bound makes $R$ and $f'(\xi)$ nonsingular, so its unique fixed point is a nonsingular zero. Whole-box Hessian or directed-interval bounds prove these inequalities, while positive interval pivots or principal minors prove positivity.

For a matching primal, pad the equal-Gram factors $\mathsf C_{\rm qry},\mathsf C_{\rm out}$ to a common row dimension. With the Moore--Penrose inverse, set $\mathsf V=\mathsf C_{\rm out}\mathsf C_{\rm qry}^+$, $\Pi_{\rm qry}=\mathsf V^T\mathsf V$ and $\Pi_{\rm out}=\mathsf V\mathsf V^T$. Equal Grams make $\mathsf V$ a partial isometry, so
\begin{equation}\label{eq:unitaryextension}
 \mathsf{U}(\mathsf{C}_{\rm qry},\mathsf{C}_{\rm out})=\begin{pmatrix}\mathsf{V}&I-\Pi_{\rm out}\\I-\Pi_{\rm qry}&-\mathsf{V}^T\end{pmatrix}
\end{equation}
is orthogonal and sends $(\mathsf{C}_{{\rm qry},x},0)$ to $(\mathsf{C}_{{\rm out},x},0)$ for every order. Extend the output-label measurement to unused directions and use the charged phase encoding. This proves finite ordinary-XOR attainability. The eight-object recovery certificate specifies the dimensions and output-projector ranks. The corresponding dense unitaries are not explicitly constructed or simulated.

\section{Physical lift and the family upper certificate}\label{app:family}
We derive constraints satisfied by every allowed family circuit. Appendix~\ref{app:branches} attains their upper bound with finite operations. Here $p$ is the dual bound, $L,\boldsymbol Q$ are transported columns, $X,Y$ group coordinates and $F$ is the success numerator.

\subsection{Expansion of the pair-register circuit}
Put
$m=n-1,\quad a_G=2(n-2),\quad b_G=2n,\quad
 \kappa_0=2(n+1)/3,\quad \mu_2=(n^2-1)/3,
 \quad \zeta=2(n-1)(n-2)/3.$
The label permutation representation is $V_n=\mathbf1\oplus S^{(n-1,1)}$. Its tensor square is the pair-matrix representation. The identity $S^{(n-1,1)}\otimes S^{(n-1,1)}=S^{(n)}\oplus S^{(n-1,1)}\oplus S^{(n-2,2)}\oplus S^{(n-2,1,1)}$ gives the decomposition in Section~\ref{sec:family} for every $n\ge4$. The constituent dimensions sum to $2+3(n-1)+n(n-3)/2+(n-1)(n-2)/2=n^2$. All symmetric-group irreducibles here have absolutely irreducible real orthogonal models, so their multiplicity spaces give the stated orthogonal commutant. Independent exact character calculations at $n=7,8,9,12$ check the formula without replacing this tensor-product argument.

Write $\mathcal T_U,\mathcal T_V$ for the scalar Cayley blocks and $T_U,T_V$ for the
standard blocks, with their signs applied in the prescribed order.
Starting from the unnormalized $aI+b(J-I)$, define
\begin{equation}\label{eq:physicalcoords}
 \binom\alpha\beta=\mathcal T_U\binom a0,\qquad
 \gamma=\frac bn(T_U\hat e-\sigma_U\hat e),\qquad s=\sigma_Ub,
 \quad \hat e=(0,0,1)^T.
\end{equation}
Then set
\begin{equation}\label{eq:outputcoords}
\begin{aligned}
 \binom{A}{b_{\rm out}}&=\mathcal T_V\binom\alpha{s+\kappa_0\gamma_2},&
 L&=T_V(\gamma_0,0,\beta/n)^T,\\
 \boldsymbol{Q}&=T_V(0,\gamma_2/a_G,3\gamma_1/b_G)^T,&
 (c,g_0,g_1,g_2)&=(\sigma_V\beta,\gamma_0,\sigma_V\gamma_1,\gamma_2).
\end{aligned}
\end{equation}
Both columns use the same $T_V$. Their cross inner product is
therefore constrained together with their two norms.

At the canonical order $x_i=i$, put $d_i=2i-n-1$, $h_i=d_i^2-\mu_2$. Off the diagonal define $E_{ij}=\operatorname{sgn}(i-j)$, $K^{\rm ord}_{ij}=E_{ij}(d_i+d_j)$ and $Z_{ij}=|d_i-d_j|$, all zero on the diagonal. Their row sums are $d_i$, $3h_i/2$ and $m\kappa_0+h_i/2$, respectively. Hence
$\mathscr E=E-T(d/n),\quad
 \mathscr K=K^{\rm ord}-T(3h/(2n)),\quad
 \mathscr Z=Z-\kappa_0(J-I)-S(h/a_G)$
have zero row and column sums. The first two are antisymmetric and the last symmetric. The first query gives $aI+bE$, with $bE=T(bd/n)+b\mathscr E$. Apply the scalar, standard and residual actions of $U$ using \eqref{eq:physicalcoords}. The second query preserves the diagonal and multiplies each off-diagonal entry by $E_{ij}$. Since $E_{ij}^2=1$, $E_{ij}(d_i+d_j)=K^{\rm ord}_{ij}$ and $E_{ij}(d_i-d_j)=Z_{ij}$, the queried states are
\begin{align}
 U\mathcal D_x(aI+b(J-I))
  & =\alpha I+\beta(J-I)+sE+\diag(\gamma_0d)
      +S(\gamma_1d)+T(\gamma_2d),\\
 \mathcal D_xU\mathcal D_x(aI+b(J-I))
  & =\alpha I+s(J-I)+\beta E+\gamma_0\diag(d)
      +\gamma_1K^{\rm ord}+\gamma_2Z.
\end{align}
The second line has scalar input $(\alpha,s+\kappa_0\gamma_2)$, standard columns $(\gamma_0,0,\beta/n)^T$ on $d$ and $(0,\gamma_2/a_G,3\gamma_1/b_G)^T$ on $h$, and residuals $\beta\mathscr E+\gamma_1\mathscr K+\gamma_2\mathscr Z$. Thus $V$ gives \eqref{eq:outputcoords} and
\begin{equation}\label{eq:fullreadout}
\begin{aligned}
 M={}&AI+b_{\rm out}(J-I)+\diag(L_0d+\boldsymbol{Q}_0h)
       +S(L_1d+\boldsymbol{Q}_1h)+T(L_2d+\boldsymbol{Q}_2h)\\
    &+g_2\mathscr Z+c\mathscr E+g_1\mathscr K.
\end{aligned}
\end{equation}

Metric preservation in \eqref{eq:physicalcoords} and
\eqref{eq:outputcoords} gives five quadratic equalities
\begin{equation}\label{eq:five}
\begin{aligned}
 e_1&=g_0^2+a_Gg_1^2+b_Gg_2^2+4sg_2=0,\\
 e_2&=A^2+mb_{\rm out}^2-\alpha^2-m(s+\kappa_0g_2)^2=0,\\
 e_3&=L^TG_3L-g_0^2-2c^2/n=0,\\
 e_4&=\boldsymbol{Q}^TG_3\boldsymbol{Q}-g_2^2/a_G-9g_1^2/b_G=0,\\
 e_5&=L^TG_3\boldsymbol{Q}-3cg_1/n=0.
\end{aligned}
\end{equation}
The first follows by expanding $\|\gamma+(s/n)\hat e\|_{G_3}^2=2s^2/n$, which is an identity for all $b$ and $s$. The second is the scalar $V$ norm equality. The last three
are the entries of one preserved two-column Gram matrix. Since
$\sigma_V^2=1$, both the norm and cross terms have the signs in
\eqref{eq:five}. The common squared norm is
\begin{equation}\label{eq:familyh}
 \mathcal{N}=n(a^2+mb^2)=n(\alpha^2+mc^2+ms^2)>0.
\end{equation}

\subsection{Minimum-row readout}
Put $Y=(A,L_0,\boldsymbol{Q}_0)^T$, $X=(c,g_1,g_2,b_{\rm out},L_1,\boldsymbol{Q}_1,L_2,\boldsymbol{Q}_2)^T$,
$\boldsymbol{b}=(1,-m,\zeta)^T$, and
\begin{equation}\label{eq:readoutR}
 R=\begin{pmatrix}
 -1/n&2m/n&-4/3&1&-m&\zeta&-m&\zeta\\
 1/n&-1&1&0&1&0&-1&0\\
 0&3/(2n)&-1/a_G&0&0&1&0&-1
 \end{pmatrix}.
\end{equation}
Let $\omega_4=4n(n^2-1)(n^2-4)/45$ and
\begin{equation}\label{eq:familyH}
 H_n=\begin{pmatrix}
 m&m&-\zeta\\
 m&n\mu_2-m^2&m\zeta\\
 -\zeta&m\zeta&\omega_4-\zeta^2
 \end{pmatrix}.
\end{equation}
At the minimum row, $M_{11}=\boldsymbol{b}^TY$ and
$M_{1j}=\rho_0+\rho_1d_j+\rho_2h_j$ for $j>1$, with $\rho=RX$.
The centered-power sums give
$H_n=\sum_{j=2}^n(1,d_j,h_j)^T(1,d_j,h_j)$.
It follows that the exact minimum success is $F/\mathcal{N}$, where
\begin{equation}\label{eq:familyF}
 F=(\boldsymbol{b}^TY)^2+X^TR^TH_nRX.
\end{equation}
 For rational circuit parameters, this formula is rational.

\subsection{A global quadratic identity}
Let $p,m_1,m_2,m_4,m_5$ be real certificate parameters and define
\begin{equation}
 J_m=\begin{pmatrix}m_1&m_5/2\\m_5/2&m_4\end{pmatrix},\quad
 D_m=\diag(m_2,J_m)-\boldsymbol{b}\boldsymbol{b}^T,
\end{equation}
\begin{equation}
 T_m=\begin{pmatrix}
 pnm-2m_1/n&-3m_5/(2n)\\
 -3m_5/(2n)&a_Gm_1-9m_4/b_G
 \end{pmatrix},\qquad
 d_m=b_Gm_1-m\kappa_0^2m_2-m_4/a_G,
\end{equation}
$M_8=\diag(T_m,d_m,m\,m_2,a_GJ_m,b_GJ_m)-R^TH_nR.$
With $e_z=(0,0,1,0,0,0,0,0)^T$, put
\begin{equation}\label{eq:B9}
 \mathsf S_9=\begin{pmatrix}
 m(np-m_2)&(2m_1-m\kappa_0m_2)e_z^T\\
 (2m_1-m\kappa_0m_2)e_z&M_8
 \end{pmatrix}.
\end{equation}
Direct expansion in all fourteen free coordinates gives
\begin{equation}\label{eq:globalidentity}
\begin{aligned}
 p\mathcal{N}-F+m_1(e_1+e_3)+m_2e_2+m_4e_4+m_5e_5
 ={}&(np-m_2)\alpha^2+Y^TD_mY\\
 &+(s,X)^T\mathsf S_9(s,X).
\end{aligned}
\end{equation}
Therefore $np-m_2\ge0$, $D_m\succeq0$, and $\mathsf S_9\succeq0$
imply $F/\mathcal{N}\le p$ on every physical circuit. The identity holds in unconstrained coordinates for all parameter values. The next appendix supplies its signs and a matching finite circuit uniformly in $n$.

\section{Algebraic branches and finite family attainments}\label{app:branches}
Two rational systems define the continuous family roots. The family certificate in the verification archive specifies their predictors, preconditioners and parameter-cell boxes. Here $k$ is a local scale, $G$ a squared amplitude, $\kappa_{\rm ctr}$ the contraction bound and $\mathbf v_{\rm tar}$ a target column.

\subsection{Two reduced systems}
Use $t=1/n$, $\beta_t=1-t$, $d=1-2t$, $k=2(1+t)/3$, and
$\bar z=2\beta_td/3$. Retain $m,a_G,b_G,\kappa_0$ from
Appendix~\ref{app:family}. The following scalings are exact,
\begin{equation}\label{eq:scalings}
\begin{aligned}
 s&=S/t,&c&=\chi /t,&g_1&=j,&g_2&=z,\\
 L_0&=\sqrt m E,&\boldsymbol{Q}_0&=\sqrt m\,t\upsilon ,&L_1&=\xi_X,&\boldsymbol{Q}_1&=t\xi_Y,\\
 L_2&=-d\xi_X,&\boldsymbol{Q}_2&=-dt\xi_Y,&A&=\sqrt m\,a_0/t,\\
 p&=t\pi ,&m_1&=M_1/t^2,&m_4&=M_4/t^4,&m_5&=M_5/t^3.
\end{aligned}
\end{equation}
The scaled root coordinates below are local to this appendix. Define
\begin{equation}\label{eq:scaledconstraints}
\begin{aligned}
 G_c&=\beta_tE^2+4\beta_td\xi_X^2-2\chi ^2,\\
 q_1&=G_c+2dj^2+2z^2+4Sz,\\
 q_4&=\beta_t\upsilon ^2+4\beta_td\xi_Y^2-z^2/(2d)-9j^2/2,\\
 q_5&=\beta_tE\upsilon +4\beta_td\xi_X\xi_Y-3\chi j.
\end{aligned}
\end{equation}
Also put
\begin{equation}\label{eq:sigmavector}
 \sigma=\begin{pmatrix}
 -\chi +2\beta_tj-4z/3-2\beta_t\xi_X+4\beta_td\xi_Y/3\\
 \chi -j+z+2\beta_t\xi_X\\
 3j/2-z/(2d)+2\beta_t\xi_Y
 \end{pmatrix},\qquad
 H_c=t^3\diag(1,1,t)H_{1/t}\diag(1,1,t).
\end{equation}
After rational cancellation, $H_c$ is polynomial in $t$ and is
defined at $t=0$. Thus the systems below have finite endpoint values.

For the tail, use $v_T=(S,\chi ,j,z,E,\upsilon ,\xi_X,\xi_Y)$ and set
\begin{equation}\label{eq:tailsystem}
\begin{aligned}
 a_0&=-S-kz,& \mathcal{N}_T&=\beta_t(S^2+\chi ^2),\\
 f_T&=\beta_t(-S-kz-\beta_tE+\bar z\upsilon )^2+\sigma^TH_c\sigma,\\
 \mathcal L_T&=\pi \mathcal{N}_T-f_T+M_1q_1+M_4q_4+M_5q_5.
\end{aligned}
\end{equation}
Take the eight coordinates of
$\tfrac12\nabla^2_{v_T}\mathcal L_T\,v_T$, append $q_1,q_4,q_5$,
multiply each equation by the positive $d$, and then set $z=1$.
This defines eleven equations in
$(\pi ,M_1,M_4,M_5,S,\chi ,j,E,\upsilon ,\xi_X,\xi_Y)$. The $z$ stationarity row is
retained before fixing its coordinate. The selected root defines
$P_T(t)$ on $[0,1/23]$.

For the head, use $v_H=(a_0,\chi ,j,z,E,\upsilon ,\xi_X,\xi_Y)$ and set
\begin{equation}\label{eq:headsystem}
\begin{aligned}
 \mathcal{N}_H&=\beta_t\{a_0^2+\chi ^2+(k/2)(G_c+2dj^2+2z^2)-k^2z^2\},\\
 f_H&=\beta_t(a_0-\beta_tE+\bar z\upsilon )^2+\sigma^TH_c\sigma,\\
 \mathcal L_H&=\pi \mathcal{N}_H-f_H+M_4q_4+M_5q_5.
\end{aligned}
\end{equation}
The eight coordinates of
$\tfrac12\nabla^2_{v_H}\mathcal L_H\,v_H$, together with $q_4,q_5$,
are multiplied by $d$ before setting $z=1$. The result is ten
equations in $(\pi ,M_4,M_5,a_0,\chi ,j,E,\upsilon ,\xi_X,\xi_Y)$. Their selected root
defines $P_H(t)$ on $[1/24,1/7]$.
The following identity relates the reduced-system norm to the norm of the physical realization. For a scaled square $A_\alpha$,
\begin{equation}\label{eq:normidentity}
 \beta_t(A_\alpha+\chi ^2+S^2)+\frac{1-t^2}{3}q_1
 +\beta_t\{a_0^2-A_\alpha-(S+kz)^2\}=\mathcal{N}_H.
\end{equation}
Once the constraints and nonnegative squares are restored, $\mathcal{N}_H$
is the scaled norm from \eqref{eq:familyh}.

\subsection{Continuous root selection}
The data divide the two parameter domains into the cells
\begin{equation}\label{eq:cells}
\begin{aligned}
 I_i^T&=[i/1840,(i+1)/1840],&&0\le i<80,\\
 I_i^H&=[1/24+17i/26880,1/24+17(i+1)/26880],&&0\le i<160.
\end{aligned}
\end{equation}
Each cell has a quadratic predictor $a(t)$ with denominator $10^{40}$, a rational preconditioner $R_{\rm pre}$ and nominal radius $R_0=10^{-7}$. For either system $\Phi(t,\xi)$, whole-cell bounds give
\begin{equation}\label{eq:tube}
 \kappa_{\rm ctr}=\sup\|I-R_{\rm pre}J_\Phi\|_\infty<13/1000,\qquad
 \varepsilon_0+\kappa_{\rm ctr}R_0<R_0,\qquad
 \frac{\varepsilon_0}{1-\kappa_{\rm ctr}}<10^{-9},
\end{equation}
where $\varepsilon_0=\sup\|R_{\rm pre}\Phi(t,a(t))\|_\infty$. The Jacobian bound covers the cell and nominal $10^{-7}$ tube, and the residual bound covers its predictor. The last inequality gives the posterior radius. Exact Bernstein bounds on $[0,1]$ and outward rational intervals enclose the polynomials and remaining operations.

For each $t$, contraction gives one nonsingular zero. At shared endpoints,
$\|a_i-a_{i+1}\|_\infty+\varepsilon_i/(1-\kappa_i)<10^{-7}$ places one posterior enclosure inside the neighboring nominal tube. Uniqueness joins adjacent roots into continuous branches, and nonsingularity gives the implicit function theorem. To bound derivatives, write the exact quadratic Taylor expansion about $a(t)$, with constant preconditioner $R$, $e=x-a$, Hessian $H(t)$, $\|e\|_\infty\le h$ and contraction $\kappa_c<1$. Differentiate along the predictor and use the Neumann inverse bound to obtain
\begin{equation}\label{eq:rootderivativebound}
\|e'\|_\infty\le
\frac{\|(R\Phi(t,a))'\|_\infty+h\|(R\Phi_x(t,a))'\|_\infty
+\tfrac12h^2\|(RH(t))'\|_\infty}{1-\kappa_c}.
\end{equation}
Whole-cell Bernstein bounds enclose the true-root derivative. The nominal radius suffices here, while positivity uses the posterior radius.

\subsection{The branch transition}
On the tail, with $z=1$, define
\begin{equation}
 a_T=-S-k,\qquad m_{2,T}=\frac{a_T-\beta_tE+\bar z\upsilon }{a_T},
 \qquad \gamma_T=\pi -m_{2,T},\qquad G_T=-2dj^2-2-4S.
\end{equation}
On the head, define
$G_H=G_c,\quad S_H=-(G_H+2dj^2+2)/4,\quad
 \Delta_H=a_0^2-(S_H+k)^2,
 \quad M_{1,H}=\pi (1-t^2)/3,\quad m_{2,H}=\pi .$
On the selected true-root branches over the whole tail interval
$[0,1/23]$ and whole head interval $[1/24,1/7]$, respectively,
the certified derivatives satisfy
\begin{equation}\label{eq:phasemonotone}
 -10<\gamma_T'(t)<-47/10,\qquad
 8/5<\Delta_H'(t)<19/5.
\end{equation}
The tail and head endpoint values satisfy
\begin{equation}
\begin{aligned}
 \gamma_T(1/24)&\in[0.008818369991,0.008818370915],\\
 \gamma_T(1/23)&\in[-0.000307992259,-0.000307967287],\\
 \Delta_H(1/24)&\in[-0.006507424333,-0.006507417716],\\
 \Delta_H(1/23)&\in[0.000229593416,0.000229595934].
\end{aligned}
\end{equation}
At $434161/10^7$ and $434162/10^7$, the tail slack lies in $(10^{-7},4\cdot10^{-7})$ and $(-4\cdot10^{-7},-10^{-7})$, respectively. These signs and its strict derivative locate one zero. To identify the head at that zero, use tail stationarity,
$\beta_t\gamma_TS+2M_1-\beta_tkm_{2,T}=0.$
At $\gamma_T=0$ this gives the head multipliers. Substituting $a_0=-S-kz$ gives $\mathcal N_H=\mathcal N_T+(1-t^2)q_1/3$ and $f_H=f_T$, with $q_4,q_5$ unchanged, so the tail root solves the head system. On the switch interval it is within $1/500000$ of the head predictor. An enlarged tube of radius $1/10000$ has contraction below $11/1000$, so uniqueness identifies the selected head root and $q_1=0$ gives $\Delta_H=0$. The derivatives in \eqref{eq:phasemonotone} imply $\gamma_T\ge0$ on $(0,\tau]$ and $\Delta_H\ge0$ on $[\tau,1/7]$. Appending $a_T\gamma_T=0$ gives a nonsingular augmented Jacobian because $a_T\gamma_T'(\tau)\ne0$, so the transition is algebraic. Since $23<\tau^{-1}<24$, the integer split is \eqref{eq:familyvalue}.

\subsection{Full-space positivity}
We next connect the selected roots to the global identity
\eqref{eq:globalidentity}. For $t>0$, apply the invertible
congruences
\begin{equation}
D_{\rm sc}=\diag(1,t,t^2)D_m\diag(1,t,t^2),\qquad
 \mathsf S_{\rm sc}=t^3W\mathsf S_9W,
 \quad W=\diag(1/t,1/t,1,1,1/t,1,t,1,t).
\end{equation}
Write $\gamma=\pi-m_2$ for the scalar slack, so that $\gamma=\gamma_T$ on the tail and $\gamma=0$ on the head. Let $\mathsf{C}_8$ be the lower eight-dimensional block of $\mathsf S_{\rm sc}$ and put
$R_8=\mathsf{C}_8-\beta_t\gamma S^2e_ze_z^T$. The stationarity equation
gives the division-free decomposition
\begin{equation}\label{eq:rankchange}
 (u_s,u_{\rm off})^T\mathsf S_{\rm sc}(u_s,u_{\rm off})
 =\beta_t\gamma(u_s-Su_z)^2+u_{\rm off}^TR_8u_{\rm off}.
\end{equation}
This identity remains valid at $\gamma=0$.

Before setting $z=1$, the tail choice $m_2=(a_0-\beta_tE+\bar z\upsilon)/a_0$ makes the first entry of $D_{\rm sc}(a_0,E,\upsilon)^T$ zero. Its other entries are $(\partial_E\mathcal L_T,\partial_{\upsilon}\mathcal L_T)/(2\beta_t)$. On the head, $m_2=\pi$ and $M_1=(1-t^2)\pi/3$ give all entries as $(\partial_{a_0}\mathcal L_H,\partial_E\mathcal L_H,\partial_{\upsilon}\mathcal L_H)/(2\beta_t)$. Thus stationarity gives a kernel with positive last coordinate. The leading two-dimensional block is positive on every posterior root enclosure of radius below $10^{-9}$, so the Schur complement vanishes and $D_{\rm sc}\succeq0$ has rank two.

The complete stationary equations restore the row kernel
$u_0=(\chi ,j,1,0,\xi_X,\xi_Y,-d\xi_X,-d\xi_Y)^T.$
Write $r=R_8u_0$. Let $\sigma_0$, $\sigma_1$ and $\sigma_2$ be the first, second and third entries of $\sigma$ in \eqref{eq:sigmavector}. Direct expansion gives $r_4=0$ from $\sigma_0+\sigma_1-(2d/3)\sigma_2=0$ and the first row $t^2\beta_t(1,1,-2d/3)$ of $H_c$. Paired symmetric and antisymmetric metric terms give $r_5+r_7=0$ and $r_6+r_8=0$. The other entries are
\[
 r_1=\tfrac12\partial_\chi\mathcal L,\quad
 r_2=\tfrac12\partial_j\mathcal L,\quad
 r_5=\frac{\partial_{\xi_X}\mathcal L}{4\beta_t},\quad
 r_6=\frac{\partial_{\xi_Y}\mathcal L}{4\beta_t}.
\]
For the selected Lagrangian $\mathcal L$, $r_3=(\partial_z\mathcal L_T-S\partial_S\mathcal L_T)/2$ on the tail and $r_3=\partial_z\mathcal L_H/2$ on the head, before substituting $z=1$. Stationarity sets these five entries to zero, and the structural identities set the other three to zero. The programs verify the rational identities.

A second kernel follows from
\begin{equation}\label{eq:secondkernel}
\begin{aligned}
 \mathsf S_9(0,0,0,0,2A,L_0,\boldsymbol{Q}_0,L_0,\boldsymbol{Q}_0)^T
 ={}&(0,0,0,0,2m d_A,\\
     &a_Gd_L,a_Gd_{\boldsymbol Q},b_Gd_L,b_Gd_{\boldsymbol Q})^T,
\end{aligned}
\end{equation}
where $(d_A,d_L,d_{\boldsymbol Q})^T=D_m(A,L_0,\boldsymbol{Q}_0)^T$.
After the scaling congruence, multiplication by $R_8$ of the vector below gives $(0,0,0,\beta_t d_0/a_0,d d_E/a_0,d d_{\upsilon}/a_0,d_E/a_0,d_{\upsilon}/a_0)^T$, where $(d_0,d_E,d_{\upsilon})^T=D_{\rm sc}(a_0,E,\upsilon)^T$. The rank-one correction in $R_8$ does not affect this vector because its $z$ coordinate is zero. The diagonal kernel therefore gives $R_8u_p=0$ with
\begin{equation}
u_p=(0,0,0,1,E/(2a_0),\upsilon /(2a_0),E/(2a_0),\upsilon /(2a_0))^T.
\end{equation}
The last two coordinates of $u_0,u_p$ have determinant
$-d(\xi_X\upsilon -\xi_YE)/(2a_0)\ne0$. The leading six-dimensional block of $R_8$ is positive on those same
posterior root enclosures, so its remaining two-dimensional Schur
complement is exactly zero at the root. Thus $R_8\succeq0$ with rank six.

The tail has $\gamma_T\ge0$ on its selected domain and the head has $\gamma=0$. Hence \eqref{eq:rankchange} gives $\mathsf S_9\succeq0$, with $np-m_2=\gamma\ge0$. Since all sign tuples satisfy the forward equalities, \eqref{eq:globalidentity} bounds every family circuit.

\subsection{Finite realization of the family optimum}
The sixteen charts cover the commutant's success probabilities up to closure. The first query has zero symmetric residual, so its sign in $U$ is irrelevant. A global sign fixes that of $V$ without affecting probabilities. A scalar rotation can match any orthogonal map on its one input vector, and finite Cayley rotations are dense. The standard Cayley block and $\eta$ cover both components of $O(3)$ up to closure, and $\sigma$ supplies the antisymmetric sign.

Restore the coordinates by \eqref{eq:scalings}, take $b_{\rm out}=0$, $z=1$,
and set
$g_0=\sqrt{G/t},\qquad \alpha=\sqrt m\sqrt\Delta/t,$
where $(G,\Delta)=(G_T,0)$ on the tail and
$(G_H,\Delta_H)$ on the head. All five equations in
\eqref{eq:five} hold. The squared norm is
$\mathcal{N}=t^{-4}\beta_t(\Delta+\chi ^2+S^2)>0.$
Multiplying the eight stationary equations by their coordinates gives
$\pi \mathcal{N}_T=f_T$ or $\pi \mathcal{N}_H=f_H$. Therefore $F/\mathcal{N}=t\pi $.

It remains to realize that lifted point with finite allowed operations.
Let $a_{\rm in}=\sqrt{\alpha^2+mc^2}$ and $b_{\rm in}=s$.
In the unnormalized $z=1$ gauge, choose
\begin{equation}\label{eq:finiteparameters}
 u_2=-\frac{mc}{a_{\rm in}+\alpha},\qquad
 v_2=\frac{m(s+\kappa_0)}{A+\alpha},\qquad
 (u_{01},u_{02},u_{12})=(0,g_0,a_Gg_1).
\end{equation}
On the posterior root enclosures of radius below $10^{-9}$ in every
tail and head cell, the exact checks give
\begin{equation}\label{eq:attainmentinequalities}
\begin{gathered}
 G>17/1000,\quad a_0>41/100,\quad \chi >74/100,\quad
 \upsilon >238/100,\quad \xi_X>20/100,\\
 \pi >0,\quad m_2>0,\quad S<0,\quad j<0,\quad
 S+k<-19/100,\quad \xi_X\upsilon -\xi_YE>41/100.
\end{gathered}
\end{equation}
These bounds give $G>0$, $s<0$ and positive $a_{\rm in}+\alpha,A+\alpha$ on each branch. The transition proof gives $\Delta\ge0$ on selected domains, so $\alpha$ is real for $t>0$. Since $d\ge5/7$ and $\beta_t\ge6/7$, all rational denominators are positive. The parameters use $\eta_U=-1$ before the Cayley factor. At $\alpha=0$ both scalar parameters are $-\sqrt m$, hence finite.

The whitened, rescaled source and target columns are
\begin{equation}
\begin{aligned}
 f&=(\sqrt G,0,\sqrt2\chi )^T,&
 g&=(0,1/\sqrt{2d},3j/\sqrt2)^T,\\
 \ell&=(\sqrt{\beta_t}E,\sqrt{2d}\,\xi_X,-\sqrt{2}\,d\xi_X)^T,&
 \mathbf v_{\rm tar}&=(\sqrt{\beta_t}\upsilon ,\sqrt{2d}\,\xi_Y,-\sqrt{2}\,d\xi_Y)^T.
\end{aligned}
\end{equation}
The full two-column Gram equality gives the proper orthogonal map
$O_V=[\ell,\mathbf v_{\rm tar},\ell\times \mathbf v_{\rm tar}]\,[f,g,f\times g]^{-1}.$
The posterior enclosures give a common Gram determinant above $65/100$ and $1+\tr O_V>38/100$. Since $\det(I+O_V)=2(1+\tr O_V)$ for a proper orthogonal matrix, its inverse Cayley chart is finite. Set
$T_V=G_3^{-1/2}O_VG_3^{1/2},\qquad
 K_V=G_3(T_V-I)(T_V+I)^{-1}.$
The skew entries of this shared map give the final standard parameters with $\eta_V=+1$. Both residual signs are $+1$. Normalize by $\sqrt{\mathcal N}$ and use two charged implementations of \eqref{eq:pairphase}. Success is $t\pi$, equal to the upper bound.

The continuous certificates give $7.43703431<P_T(0)<7.43703432$ and $-21<P_T'(t)<-18$ on $[0,1/23]$. Integration proves \eqref{eq:familyerror}. The archived family and endpoint checkers verify the endpoint, transition, full slacks and shared final map. 

At seven objects, one optimal family circuit starts with diagonal and off-diagonal amplitudes $0.2707478521$ and $-0.1076667261$. It uses $u_2=-1.3829494983$, $v_2=-0.4522615612$, standard triples $(0,1.0950425397,-10.7514572121)$ and $(-3.1645981226,-5.2298688384,13.9023316964)$, and signs $(-1,+1,+1,+1)$. Each triple lists $(\Omega_{01},\Omega_{02},\Omega_{12})$ in $(z_0,z_1,z_2)$ order. These decimals approximate the exact algebraic circuit, with success approximately $0.7118121160$.

\section{The uniform unrestricted upper certificate}\label{app:uniform}
Unlike the family bounds, this complete-model certificate applies to every algorithm in Section~\ref{sec:model}, with objective $c_U/n$. The uniform upper certificate in the verification archive uses the signed model catalogue of Appendix~\ref{app:symmetry}. Here $k$ counts marked free labels and $L$ is their subset.

\subsection{Shared-subset positivity}
For a degree-at-most-$d$ algebra with anchors $A$, extend a coefficient vector $s$ to a pattern function $\widehat S_s(\alpha,\beta)$ by signed canonicalization, with negative-stabilizer patterns zero. A $k$-flag is a feature with $k$ distinguished free labels modulo permutations of unmarked labels. Anchors and marks are individually fixed, also for the edge-query algebra. Negative unmarked stabilizers give zero rows. For flags $f,g$ with no negative unmarked stabilizer, define
\begin{equation}\label{eq:flagmatrix}
 \mathsf F_k(s)_{f,g}=\sum_{L\subseteq[k]}(-1)^{k-|L|}
        \widehat S_s(f\mathbin{\sqcup_L}g).
\end{equation}
The copies in $f\mathbin{\sqcup_L}g$ share the anchors and precisely the marks in $L$, with disjoint remaining tails and orientation signs retained. For each actual $\mathcal J\subseteq[n]\setminus A$, the signed selector $\mathsf T_{\mathcal J}$ has a zero row unless the feature contains $\mathcal J$. Otherwise it marks those labels in increasing order and inserts the canonicalization sign. For a feature pair only subsets of $I=(\operatorname{supp}\alpha\cap\operatorname{supp}\beta)\setminus A$ contribute. The coefficient identifying precisely $L\subseteq I$ is
\begin{equation}
 \sum_{L\subseteq \mathcal{J}\subseteq I}(-1)^{|\mathcal{J}|-|L|}
 =\begin{cases}1,&L=I,\\0,&L\ne I.
 \end{cases}
\end{equation}
Thus, entry by entry,
\begin{equation}\label{eq:flagidentity}
 S(n)=\sum_{\substack{\mathcal{J}\subseteq[n]\setminus A\\|\mathcal{J}|\le2d}}
                    \mathsf T_\mathcal{J}\mathsf F_{|\mathcal{J}|}(s)\mathsf T_\mathcal{J}^T.
\end{equation}
Each labelled subset contributes one term with coefficient one.
Odd flags contribute zero because their unmarked automorphism acts
on one disjoint tail while preserving the other.

At $n=7$, some auxiliary disjoint-tail patterns in
\eqref{eq:flagmatrix} use eight or nine abstract labels. Every such
term in these maps has a negative stabilizer and is zero. The
nonzero patterns already occur in the seven-label catalogue. Hence
\eqref{eq:flagidentity} applies at $n=7$ as well as at larger sizes.
Positivity of all finite $\mathsf F_k(s)$ implies positivity of the entire
kernel. This is a sufficient condition. Symmetric sums of squares can
likewise have descriptions whose size depends on the degree rather than
the population size~\cite{RaymondEtAl2018}. The flags here are literal
shared-label features, rather than the asymptotic densities of flag
algebras~\cite{Razborov2007}.

\subsection{The rational dual}
Use the coefficient order of the global quadratic catalogue in the signed
model certificate. Set $d_{\rm den}=2\cdot10^{15}$ and $\bar z=v/d_{\rm den}$.
The uniform upper certificate specifies the 34 integer coordinates of $v$.
Let $\bar\gamma=5127403871048098/d_{\rm den}$ and
$\overline{\mathbf h}=(c_U,c_U,\bar\gamma)$, where $c_U$ is the rational constant in
\eqref{eq:constants}. The dual coordinates are
$\theta=c_U/n,\qquad \mathbf{h}=\overline{\mathbf h}/n,\qquad z=\bar z/n.$
Their complete scaled query and output slacks are
\begin{equation}\label{eq:uniformslacks}
 \bar a_0=\overline{\mathbf h}-C_0\bar z,\qquad
 \bar a_e=E_1\overline{\mathbf h}-C_e\bar z,\qquad
 \bar s=E_2\bar z-\bar f,
\end{equation}
where $\bar f=nf(n,1)$ is the exact scaled minimum objective
from \eqref{eq:momentformula}. All entries are rational and
independent of $n$ in these signed coordinates. 

The projector support is checked by $\bar Z\kappa_{abc}=0$.
It suffices to check all rows against one fixed triangle using at
most seven labels, because a quadratic row uses at most four other
labels. Signed covariance then gives every triangle and every size.
Self-adjointness gives support on both sides. The output moment has the same support. The output slack
$(E_2\bar z-\bar f)/n$ is therefore supported, and its positive
semidefiniteness on that support is precisely the condition needed
in weak duality. The scalar dual weight slacks are zero.

\subsection{Full positivity and weak duality}
The flag maps applied to \eqref{eq:uniformslacks} give eleven
rational matrices. The uniform upper certificate specifies every matrix entry and its positive factorization. The exact checker reconstructs each matrix and verifies its rank. The identity-query rows
certify $\overline{\mathbf h}-C_0\bar z$, the edge-query rows certify
$E_1\overline{\mathbf h}-C_e\bar z$, and the output rows certify
$E_2\bar z-\bar f$. Each full matrix has the factorization
\begin{equation}\label{eq:flagfactor}
 \mathsf F=T^TL\diag(d_1,\ldots,d_r)L^TT,
 \qquad d_j>0,
\end{equation}
where $L$ is unit lower triangular and selected columns of $T$
form the identity. The data specify all entries. Exact rational multiplication reproduces each whole matrix. The omitted columns give explicit null vectors,
which annihilate the full matrix.

Equations \eqref{eq:flagidentity} and \eqref{eq:flagfactor} prove full slack positivity for every finite $n\ge7$. Weighted adjoints and projector support preserve the physical dual pairing, so \eqref{eq:gap} bounds success by $c_U/n$. The complete model covers every admitted algorithm.

\section{Certificate data and verification}\label{app:verification}
The verification archive~\cite{WuTwoComparisonArchive} contains the exact certificates and checking programs under the MIT license. These certificates supply the concrete witnesses used in the computer-assisted proofs. The appendices give the derivations and acceptance conditions. The archive's certificate index identifies the exact inputs and checking programs for each result.

The checkers verify the signed model, polynomial roots, positivity, causal equalities, complementarity and the certificate obligations of the all-size results. Exact checks use integer and rational arithmetic, outward rational intervals and the supplied polynomial modules. Floating-point diagnostics are recorded separately. The archive contains file locations, pinned dependencies and execution instructions. A manifest records the byte length and SHA-256 of each file, and the replay checks it before and after execution.

\section{Numerical solutions of the complete model}\label{app:numerics}
Table~\ref{tab:smallnumerics} gives unreduced numerical solutions, with only the triangle kernel removed. The exact values at sizes four through six are not certified here. Table~\ref{tab:numerical} compares selected complete-model estimates with certified family intervals. The difference uses the interval midpoint. The numerical records in the verification archive contain the solver settings, residuals and simulation results. The fit certificates specify the exact finite-node bounds and leave-one-out calculations.

\begin{table}[tb]
\centering\small
\caption{Small-size numerical success}\label{tab:smallnumerics}
\begin{tabular}{rrrrrrr}\toprule
$n$&$r=1$&$r=2$&$r=3$&$r=4$&$r=5$&$r=6$\\\midrule
3&1.000000&1.000000&1.000000&--&--&--\\
4&0.936526&1.000000&1.000000&1.000000&--&--\\
5&0.858008&0.982149&1.000000&1.000000&1.000000&--\\
6&0.783213&0.955213&0.992904&1.000000&1.000000&1.000000\\
\bottomrule\end{tabular}
\end{table}

\begin{table}[tb]
\centering\small
\caption{Numerical minimum-selection values}\label{tab:numerical}
\begin{tabular}{rrrr}\toprule
$n$&$n\Q(n,1)$, numerical&Certified $n\PW(n)$ interval&Difference, numerical\\\midrule
9&5.474051113&$[5.442054895, 5.442054897]$&0.031996217\\
23&6.596560370&$[6.581981578, 6.581981581]$&0.014578791\\
24&6.629640707&$[6.615822905, 6.615822908]$&0.013817800\\
48&7.022259436&$[7.016172415, 7.016172417]$&0.006087020\\
60&7.103536960&$[7.098706813, 7.098706815]$&0.004830146\\
100&7.235541328&$[7.232466973, 7.232466975]$&0.003074355\\
500&7.397181234&$[7.395744435, 7.395744437]$&0.001436798\\
1000&7.417637542&$[7.416365873, 7.416365875]$&0.001271668\\
\bottomrule\end{tabular}
\end{table}

The fits assume an expansion of the unrestricted scaled optimum in powers of $1/n$. The node sets are heuristic. Degrees two and three use $n=100,200,500,1000$, and degree four adds $n=40,48,49,60$. The fitted intercepts are $7.43814878611$, $7.43815014533$ and $7.43815016657$. The cubic interpolates its four nodes, so its residual is not an independent accuracy estimate. For a stored sample $s_n$, exact repaired primal and dual values $L_n,U_n$ give $\delta_n=\max\{|s_n-L_n|,|U_n-s_n|\}$, and exact intercept weights give propagated uncertainty $\sum_n|w_n|\delta_n$. The propagated bounds for degrees two, three and four are $4.458\cdot10^{-5}$, $1.04\cdot10^{-4}$ and $4.735\cdot10^{-5}$, rounded upward. These bounds do not bound the bias from truncation or the numerical flag-optimum error. The conjecture in Section~\ref{sec:uniform} depends on the assumed expansion and uncertified solver accuracy.

Floating-point solutions of the complete model for $9\le n\le60$ and $n=100,200,500,1000$ give a difference $n\Q(n,1)-n\PW(n)$ that decreases from about $0.032$ at $n=9$ to about $0.00127$ at $n=1000$ over the solved sizes. The cubic and quartic intercepts differ from the CVXOPT flag-dual value $7.4381501687$ by about $2.3\cdot10^{-8}$ and $2.2\cdot10^{-9}$, while their certified propagated node budgets are $1.04\cdot10^{-4}$ and $4.735\cdot10^{-5}$. Clarabel returns \texttt{optimal\_inaccurate} with $7.4381501016$, about $6.7\cdot10^{-8}$ below CVXOPT, and a worst block eigenvalue of $-2.1\cdot10^{-8}$. The archived leave-one-out calculations retain the underdetermined cubic cases. Extrapolation from the family data recovers its limit within about $1.7\cdot10^{-9}$.

\section{A finite certificate of strict asymptotic separation}\label{app:strict}
The limit comparison in Theorem~\ref{thm:strict} follows from one
rational feasible point at $n_0=50000$. The strict-separation certificate
in the verification archive specifies its seeds and weights. Its exact
checker uses the signed catalogue and maps of Appendix~\ref{app:symmetry}. If $v_0$ is its
exact objective, direct rational evaluation gives
\begin{equation}\label{eq:strictobjective}
 \frac{743710307}{10^8}<n_0v_0<\frac{743710308}{10^8}.
\end{equation}
Consequently, $n_0v_0-7.43703432>6875/10^8$.
The certificate records the exact rational objective.

Convert the candidate conic multipliers to rationals with forty digits after the leading decimal digit. Impose signed self-adjointness and rational weights $w_0+Nw_e=1$. Keeping $x_e$, set
\begin{equation}\label{eq:strictrepair}
\begin{aligned}
 x_0&=(w_0,w_e,0)-T_1x_e,\\
 y'&=p_o*y*p_o,\\
 g&=p*(T_0x_0+T_ex_e)*p-T_oy',\\
 y''&=y'+E_2g/n_0.
\end{aligned}
\end{equation}
The identity $T_oE_2g=n_0g$ restores the second flow, while the first line restores the first. Since $g$ is supported, output support is preserved. These operations are rational but need not preserve positivity. Mix with weight $10^{-4}$ of a uniform-output point that independently chooses identity weight $1/2$ and each edge weight $1/(2N)$ at both calls. For $w^{\rm u}=(1/2,1/(2N),0)$ its seeds are
\begin{equation}\label{eq:strictinterior}
 x_0^{\rm u}=w^{\rm u}/2,\qquad
 x_e^{\rm u}=E_1w^{\rm u}/(2N),\qquad
 y^{\rm u}=E_2\{p*(T_0x_0^{\rm u}+T_ex_e^{\rm u})*p\}/n_0.
\end{equation}
This point outputs each object with probability $1/n_0$. Mixing preserves the flows and support, and retains \eqref{eq:strictobjective}. Exact fraction-free elimination of the three regular images gives positive nonzero pivots of ranks $3,14,127$ in sizes $3,14,166$, with zero remaining entries on zero-pivot rows. The checker verifies normalized nonnegative weights, all three and 34 flow coordinates, support and objective. By Appendix~\ref{app:symmetry}, coset conjugation gives full positive matrices and the checked adjoints give their sums. Theorem~\ref{thm:model} then supplies the finite realization. The checker imports no optimizer or producer. The strict-point checker also tests deliberately corrupted inputs.

\section{Exact threshold certificates}\label{app:thresholdcert}
The two boundary certificates use the complete signed coefficient model
of Appendix~\ref{app:symmetry}. The threshold certificates in the
verification archive specify a rational primal at $n=46$ and a rational
dual at $n=45$. The primal is recovered and repaired by the method of
Appendix~\ref{app:strict}. A mixture of weight $1/10000$ with
the uniform-output point makes its seeds positive on their supports.
Exact regular-image elimination gives ranks $3,14,127$ for the
identity-query, edge-query and output seeds. Nonnegative query weights,
both full causal equalities and projector support hold exactly.
Its value $v_{46}$ satisfies
\begin{equation}\label{eq:thresholdprimal}
7.00411359<46v_{46}<7.00411360.
\end{equation}
The independent primal checker accepts integer and rational arithmetic only.

The rational dual at $n=45$
uses scaled coordinates $\bar h=45h$ and $\bar z=45z$.
Its slacks are $\bar h-C_0\bar z$, $E_1\bar h-C_e\bar z$ and
$E_2\bar z-45f(45,1)$. The independent checker verifies their exact
regular-image positivity, signed self-adjointness, projector support,
and both scalar inequalities. It computes the causal adjoints from the weighted-transpose definition. It also verifies the exact weak-duality telescoping identity on
a separately formed uniform-output point. The symbolic identity
\eqref{eq:gap} then bounds every complete-model algorithm.
For its upper value $\theta_{45}$,
\begin{equation}\label{eq:thresholddual}
6.99537459<45\theta_{45}<6.99537460.
\end{equation}
These inequalities and monotonicity prove Corollary~\ref{cor:threshold}. Both threshold checkers also test deliberately corrupted certificates.

\end{document}